\documentclass[11pt]{article}

\usepackage{cite,fullpage}
\usepackage{graphicx}
\usepackage{hyperref}

\usepackage{algorithm}
\usepackage{algorithmicx}
\usepackage[noend]{algpseudocode}
\usepackage{algpascal}

\usepackage{authblk}
\usepackage{subcaption}
\usepackage{multirow}
\usepackage{enumerate}

\usepackage{amsmath}
\usepackage{amssymb}
\usepackage{color}
\definecolor{red}{RGB}{255,0,0}
\definecolor{blue}{RGB}{0,0,255}
\definecolor{green}{RGB}{0,255,0}

\newenvironment{proof}{\noindent{\bf Proof:}}{$\square$ \vspace{3mm}}

\newcommand {\abs}[1]  {\left\vert#1\right\vert}
\newcommand {\set}[1]  {\left\{#1\right\}}

\newcommand {\defined} {\stackrel{def} {=}}

\newcommand {\bigoh}   {{\cal O}}

\newcommand {\runningtitle}[1] {\vspace{0.5ex}\noindent{\textbf{\boldmath #1:}}}

\newcommand{\ignore}[1] {}
\newcommand {\commentfig}[1] {#1}
\renewcommand{\commentfig}[1] {}

\newcommand{\LineComment}[1]{\State~\(\triangleright\)~#1}
\newcommand{\Yield}[1]{\textbf{\textrm{yield}}~#1}

\newcommand{\dir}[1]{\overrightarrow{#1}}
\newtheorem{theorem} {Theorem}
\newtheorem{lemma} {Lemma}
\newtheorem{definition}  {Definition} 
  
\newtheorem{fact} {Fact}
\newtheorem{corollary}  {Corollary}
\newtheorem{observation}  {Observation} 

\newcommand{\TT}  {{\cal T}}
\DeclareMathOperator{\hs}{hs}
\DeclareMathOperator{\bcw}{bcw}

\begin{document}

\title{Generation of weighted trees, block trees and block graphs \thanks{The first and the second author acknowledge the support of the The Scientific and Technological Research Council of Turkey T\"{U}B\.ITAK Grant number 122M452.}
}
\author[1]{T{\i}naz Ekim}
\author[2]{Mordechai Shalom}
\author[3]{Mehmet Aziz Yirik}

\affil[1]{
Department of Industrial Engineering, Bogazici University, Istanbul, Turkey
\footnote{tinaz.ekim@boun.edu.tr}
}
\affil[2]{
Peres Academic Center, Rehovot, Israel
\footnote{cmshalom@gmail.com}
}
\affil[3]{
Department of Mathematics and Computer Science, The University of Southern Denmark, Odense, Denmark
\footnote{mehmetazizyirik@gmail.com}
}

\maketitle 

\begin{abstract}
We present a general framework to generate trees every vertex of which has a non-negative weight and a color. 
The colors are used to impose certain restrictions on the weight and colors of other vertices. We first extend the enumeration algorithms of unweighted trees given in \cite{Wilf78, Wilf81} to generate weighted trees that allow zero weight. 
We avoid isomorphisms by generalizing the concept of centroids to weighted trees and then using the so-called centroid-rooted canonical weighted
trees. 
We provide a time complexity analysis of unranking algorithms and also show that the output delay complexity of enumeration is linear.
The framework can be used to generate graph classes taking advantage of their tree-based decompositions/representations. 
We demonstrate our framework by generating weighted block trees which are in one-to-one correspondence with connected block graphs. 
All connected block graphs up to 19 vertices are publicly available at \cite{HoGBlock}.

\end{abstract}

\section{Introduction}\label{sec:intro}
\runningtitle{Background}
Graphs provide a very useful tool to model many important real-life problems such as drug design, DNA sequencing, temporal reasoning in artificial intelligence, etc. \cite{Golumbic:2004:AGT:984029}. In each application, the related graph model admits local structural properties that are embodied by various graph classes. It is therefore important to solve optimization problems in graph classes defined by special structures.  Accordingly, many algorithms have been developed for various optimization problems in these classes. 
Yet, no random graph generators for these graph classes have been provided in the literature to measure the performance of exact/approximation/parameterized algorithms apart from a few exceptions.
This especially holds for the average case performance, since most of the theoretical analyzes focus on the worst case performance of the algorithms. 
This shortcoming was even mentioned as an important open task at a Dagstuhl Seminar \cite{bodlaender_et_al:DR:2014:4544}. 

An important use-case of graph generation with specified restrictions is in chemical graph theory, such as drug discovery and natural product studies \cite{yirik2021chemical}. To elucidate the structure of unknown molecules, chemists need the entire chemical space to be constructed, in other words, the isomorphism-free graph class. 
Constructing the entire graph class can be a very time-consuming task. 
Therefore, chemists aim to get a grasp of a subset of the class as the first step towards conducting in-depth analyses. 
One way to construct such an unbiased subset is the uniformly at random generation of the graph classes \cite{kerber1990cataloging}.   

It is well-known that various types of trees are used to represent structures of graphs. 
Clique-trees, block-trees, tree decompositions, clique-width decompositions, decompositions by clique-separators are some of the well-known such representations. 
Whenever a one-to-one mapping between graphs of a certain family and their tree representations can be established, the generation of the representative trees uniformly at random provides us with a tool to generate the graphs in the family uniformly at random as well. 
Motivated by the above examples, in this study, we present a generic framework to enumerate trees with some color scheme to express additional properties and demonstrate our approach to generate block graphs (with a given number of vertices).

Enumeration of rooted trees is a well-studied problem.
We will mention here a few works that are closely related to our work.
The book \cite{Wilf78} contains enumeration algorithms for numerous sets of combinatorial objects and a general framework to obtain such algorithms.
It also provides FORTRAN implementations of the algorithms.
As such, these algorithms are iterative (as opposed to recursive) since FORTRAN, at that time, did not have a direct support for recursive subroutines.
Among others, the book presents enumeration algorithms for (unlabeled, unordered) rooted trees.
Enumeration algorithms for free (i.e., unrooted) trees are presented in the paper \cite{Wilf81} by providing a one-to-one-mapping between free trees and sets of at most two rooted trees.  

The work \cite{L96} presents recursive enumeration algorithms for rooted and free trees with Pascal implementations.
It extends the results to trees with degree restrictions.
Being recursive, their implementation is, by far, simpler.
They show that the enumeration algorithm is CAT, i.e., it runs in constant amortized time.
Namely, the time needed to enumerate all the trees is proportional to the number of trees. 
In \cite{kelsen}, the author presents alternative recursive ranking and unranking algorithms for binary trees of bounded height and for height-balanced trees.

\newcommand{\bgg}{\textsc{BlockGraphGen}}

\runningtitle{Our Contribution}
In this work, we present a general framework to enumerate and unrank trees that are weighted and colored.
These are trees every vertex of which has a non-negative weight and a color. 
The colors may be used to impose certain restrictions on the weight and colors of other vertices.
We demonstrate this framework by generating block graphs and call this algorithm {\bgg}.
For this purpose, we assign colors and weights to the vertices of the block tree of a hypothetical graph on $n$ vertices and we enumerate/unrank all such trees.
Since block trees are in one-to-one correspondence with block graphs, the enumeration/unranking of block graphs follows. All non-isomorphic connected block graphs with up to 19 vertices are downloadable at \cite{HoGBlock}, the House of Graphs searchable database of "interesting graphs" \cite{HoG}.

\runningtitle{Organization of the paper}
In Section \ref{sec:prelim}, we start with a discussion on enumeration problems and introduce the terminology related to trees. 
We continue with our implementation of the algorithms in \cite{Wilf78} and \cite{Wilf81} in Section \ref{sec:Trees}. 
We present them as a set of generators involving indirect recursion. 
Taking advantage of features available today in many programming languages, our pseudo code is very simple.
Moreover, it is almost identical to its Python implementation \cite{PythonImpl}. We show that the enumeration of rooted trees and forests can be done in linear output delay time (in Theorem \ref{thm:outputdelay}) and their unranking takes $\bigoh(n^3)$ time (Theorem \ref{thm:UnrankingComplexity}). 

In Section \ref{sec:WeightedTrees}, we extend the techniques to weighted trees. For this purpose, we first extend the well-known structural properties of the centroids to the case of weighted trees that allow for zero-weight vertices.

In Section \ref{sec:BlockTrees} we represent block trees as vertex-colored weighted trees with specific restrictions. 
Then, we show that there is a one-to-one correspondence between block trees and block graphs (Theorem \ref{thm:block}). 
This enables us to enumerate block graphs using the generation algorithm {\bgg}.

In Section \ref{sec:framework}, we present a framework that generalizes the algorithms developed in the previous sections
and give references to a Python implementation of it.

Lastly, in Section \ref{sec:conclusion}, we discuss possible extensions of our work, the ingredients needed to adapt our approach to enumeration problems in other graph classes and point out promising research directions.

\section{Preliminaries}\label{sec:prelim}
To develop algorithmic solutions in networked systems, it is important to test algorithms on instances that are generated uniformly at random (u.a.r. for short) so that the experimental results are not biased. 
Two important issues make generating structured graphs u.a.r. an extremely challenging question. 
i) Forcing structural properties requires interventions to the generation process that harms randomness. 
The remedy depends very much on the required structure and how well this graph structure is captured and exploited for a generation. 
ii) To avoid the generation of the same graphs repetitively, one should perform a high number of isomorphism checks efficiently (which is a computationally very hard task \cite{iso20}) or avoid isomorphisms using powerful structural results. Generating (unlabeled) graphs u.a.r. means that, up to isomorphism, every graph in the target graph class is produced with equal probability. 

As pointed out in \cite{UEHARA2005479}, counting, generating u.a.r., and enumerating unlabeled graphs in a special class are closely related to each other.
Yet, these problems are not equivalent \cite{HararyPalmer73, Wilf78}, and among them, random generation seems to be the less investigated. 
\emph{Counting} is the task of finding the cardinality of a set $S$ of combinatorial objects. 
\emph{Enumeration} is the task of constructing all the objects in such a set $S$ without repetition, in some order $s_0, s_1, \ldots$.
The task of \emph{ranking} consists of determining the index of a given object $s \in S$ in the enumeration.
\emph{Unranking} (a.k.a) \emph{random access} is the opposite task, i.e., the task of finding the object $s_i \in S$ corresponding to a given non-negative integer $i$ (without enumerating $S$).
We will refer to this family of four problems as \emph{enumeration problems} through the text.

Clearly, the ability to count and unrank a set $S$ enables us to generate an element of $S$ u.a.r. as described in Algorithm \ref{alg:ChooseUAR}.

\alglanguage{pseudocode}

\begin{algorithm}
\caption{\textsc{Choose u.a.r.}}\label{alg:ChooseUAR}
\begin{algorithmic}[1]

\State $N \gets \abs{S}$  \Comment{Use the counting algorithm}
\State Choose an integer $i \in [0, N-1]$ u.a.r.
\State \Return the $i$-th element of $S$.  \Comment{Use the unranking algorithm}

\end{algorithmic}
\end{algorithm}

These building blocks can also be used to parallelize the task of enumerating objects of the set $S$ as suggested in Algorithm \ref{alg:ParallelEnumeration}.

\alglanguage{pseudocode}

\begin{algorithm}
\caption{\textsc{Parallel Enumeration}}\label{alg:ParallelEnumeration}
\begin{algorithmic}[1]
\Require{$P$ the number of processors available}
\State $N \gets \abs{S}$   \Comment{Use the counting algorithm}
\State $k \gets \lfloor \frac{N}{P} \rfloor$
\For {$i = 0$ to $P-1$ in parallel}
    \For{$j=i \cdot k$ to $\min \set{(i+1) \cdot k - 1, N-1}$}
        \State \Return the $j$-th element of $S$.  \Comment{Use the unranking algorithm}
    \EndFor
\EndFor
\end{algorithmic}
\end{algorithm}

We use standard terminology and notation for graphs (see, for instance \cite{D12}).
We denote by $[n]$ the set of positive integers not larger than $n$. 
Given a simple undirected graph $G$, we denote by $V(G)$ the set of vertices of $G$ and by $E(G)$ the set of the edges of $G$.
We use $\abs{G}$ as a shortcut for $\abs{V(G)}$.
We denote an edge between two vertices $u$ and $v$ as $uv$. In this work we consider only \emph{unlabeled}, \emph{unordered} graphs.
That is, the vertices of the graph do not have labels and the neighbors of a vertex do not have a particular order.

Throughout the paper, we consider both \emph{undirected} (a.k.a. \emph{unrooted}, a.k.a. \emph{free}) trees and 
\emph{directed} (a.k.a. \emph{rooted}) trees.
A directed tree is obtained from an undirected one by designating a vertex $r$ as its root, and directing the edges from the root outwards. 
Unless stated otherwise, a tree under consideration is undirected.

Let $T$ be a tree on $n$ vertices and $v$ a vertex of $T$.
The \emph{subtrees} of $v$ in $T$ are the subtrees of $T$ obtained by the removal of $v$, i.e., the connected components of $T - v$.
The \emph{heaviest subtree weight} of $v$ in $T$ denoted by $\hs_T(v)$ is the largest number of vertices in any of its subtrees in $T$.
When the tree $T$ is clear from the context, we refer the subtrees simply as the \emph{subtrees} of $v$ and we denote the heaviest subtree weight of $v$ as $\hs(v)$.
A \emph{centroid} of $T$ is a vertex with minimum heaviest subtree weight in $T$. The following fact regarding the centroids of a tree $T$ is well known \cite{Knuth68}.

\begin{fact}\label{fact:centroid}
Let $T$ be a tree on $n$ vertices.
\begin{itemize}
    \item $T$ has at most two centroids;
    \item if $T$ has two centroids then they are adjacent and the removal of the edge
joining them separates $T$ into two trees of equal sizes, 
i.e., the heaviest subtree weights of the centroids is $n/2$;
    \item if $v$ is the unique centroid of $T$ then $\hs(v) \leq (n - 1)/2$.
\end{itemize}
\end{fact}

\section{Unweighted Trees}\label{sec:Trees}
\newcommand{\FF}{{\mathcal{F}}}
\newcommand{\RT}{{\mathcal{RT}}}
\newcommand{\FT}{{\mathcal{T}}}
\newcommand{\CC}[2]{CC^{#1}_{#2}}

In this section, we present enumeration and unranking algorithms for unweighted trees.
The algorithms in Section \ref{sec:RootedTrees} and Section \ref{sec:FreeTrees} are very similar to the algorithms presented in \cite{Wilf78} and \cite{Wilf81}, respectively. 
Here, we provide an adaptation of these algorithms so that it will suit our needs in the upcoming sections. 
We present recursive algorithms for unweighted rooted trees and then algorithms for free trees that use these recursive algorithms as a subroutine. 
Technically, we enumerate all (unlabeled) trees using an adaptation of the formula counting unlabeled trees in \cite{Wilf81}. 
We avoid isomorphisms by exploiting the fact that every tree has either one or two adjacent centroids, as noted in Fact \ref{fact:centroid}. 
In this section, we prepare the technical basis for the development of the extensions of these algorithms to weighted trees in Section \ref{sec:WeightedTrees} and block trees in Section \ref{sec:BlockTrees}. 
Last but not least, we provide a time complexity analysis for the unranking algorithm 
and an output delay time complexity analysis for the enumeration algorithm.

\subsection{Rooted Trees}\label{sec:RootedTrees}
In this section, a tree is a rooted tree, and a forest is a digraph the connected components of which are rooted trees. 
For a forest $F$, we denote by $m(F)$ the maximum number of vertices in a tree of $F$ 
and by $\mu(F)$ the number of trees of $F$ having $m(F)$ vertices.
\bigbreak
Denote
\begin{itemize}
\item by $\RT(n)$ the set of trees on $n$ vertices,
\item by $\FF(n)$ the set of forests on $n$ vertices,
\item by $\FF(n,m)$ the set of forests $F \in \FF(n)$ with $m(F)=m$,
\item by $\FF^\leq (n,m) = \cup_{m' \in [m]} \FF(n,m')$ the set of forests $F \in \FF(n)$ with $m(F) \leq m$,
\item by $\FF(n,m,\mu)$ the set of forests of $F \in \FF(n,m)$ with $\mu(F)=\mu$, and
\item by $\FF^\leq (n,m,\mu) = \cup_{\mu' \in [\mu]} \FF(n,m,\mu')$ the set of forests of $F \in \FF(n,m)$ with $\mu(F) \leq \mu$.
\end{itemize}

\subsubsection{Enumeration}
\newcommand{\enumrtree}{\textsc{EnumerateRootedTrees}}
\newcommand{\enumforest}{\textsc{EnumerateForests}}

Under these definitions one can enumerate all the rooted trees on $n$ vertices as described by the pseudo code in Algorithm \ref{alg:RootedEnumeration}.
The pseudo code contains three generators that recursively invoke one another.
Namely,
$\enumrtree$(n) generates $\RT(n)$,  $\enumforest(n,m)$ that generates $\FF^\leq (n,m)$, and $\enumforest(n,m, \mu)$ that generates $\FF(n,m,\mu)$.
We adopted the yield statement from languages such as Python, C\#, PHP to return an object from the set being generated.
The return statement signals the termination of the set.

$\enumrtree$ repeatedly invokes $\enumforest$ to get one forest of $n-1$ nodes at-a-time and attaches to each forest a root to get a rooted tree on $n$ vertices.

$\enumforest(n,m)$ iterates over all possible largest tree sizes $m' \leq m$ and
number $\mu$ of largest trees.
For each such pair $m', \mu$ it yields all the forests yielded by $\enumforest(n,m', \mu)$

Finally, $\enumforest(n,m, \mu)$ iterates over all multisets with $\mu$ elements of the set $\RT(m)$ of trees on $m$ vertices ($highForest$). 
For each such multiset, it iterates over all forests on $n - \mu \cdot m$ vertices with largest tree size at of most $m-1$ ($lowForest$)
and it returns the union of the two ($highForest \cup lowForest$).

\newcommand{\alg}{\textsc{Enumerating Rooted Trees and Forests}}
\alglanguage{pseudocode}
\begin{algorithm}[H]
\caption{{\alg}}\label{alg:RootedEnumeration}
\begin{algorithmic}[1]

\Ensure{Enumerates the set $\RT(n)$ of rooted trees on $n$ vertices.}
\Function{\enumrtree}{$n$}
\For{$F \in$ \Call{\enumforest}{$n-1,n-1$}}
    \State \Yield{the tree obtained from $F$ by connecting all its roots to a new root}.
\EndFor
\EndFunction

\Statex
\Ensure{Enumerates the set $\FF^\leq(n,m)$ of rooted forests on $n$ vertices\\ 
~~~~~~~~with largest tree size of at most $m$.}
\Function{\enumforest}{$n,m$}
\If{$n=0$}
\State \Yield{the empty forest}
\State \Return
\EndIf
\If{$m=0$}
\State \Return
\EndIf
\For{$m' \in [m]$}
    \For{$\mu \in [\lfloor n/m' \rfloor]$}
        \For {$F \in$ \Call{\enumforest}{$n,m',\mu$}}
            \State \Yield{$F$}
        \EndFor
    \EndFor
\EndFor
\EndFunction

\Statex
\Ensure{Enumerates the set $\FF(n,m,\mu)$ of rooted forests on $n$ vertices\\ 
~~~~~~~~with $\mu$ largest trees, each having size $m$.}
\Function{\enumforest}{$n,m,\mu$}
    \For {$highForest \in $ multi-subsets with $\mu$ elements of $\RT(m)$}
        \For {$lowForest \in $ \Call{\enumforest}{$n-\mu \cdot m, \min(n - \mu \cdot m,m-1)$}}
            \State \Yield{$highForest \cup lowForest$}
        \EndFor
    \EndFor
\State \Return
\EndFunction

\end{algorithmic}
\end{algorithm}

Given the above definitions, the correctness of the functions $\enumrtree$ and 
$\enumforest$ are obvious.
We now give a time complexity analysis of this enumeration.
The \emph{delay} time of an algorithm that enumerates a set of objects is the maximum time elapsed between the emission of two consecutive objects (and the time elapsed until the first object is emitted).
We will consider the \emph{output delay} time complexity in which the delay time is measured as a function of the output size, i.e., the size of the emitted object.
Throughout the analysis, the \emph{size} of an object is the number of bits of its representation.
In particular, the size of a set is the sum of the sizes of its elements (as opposed to the number of its elements).
The size of a graph $(V,E)$ is $\bigoh(\abs{V}+\abs{E})$.
In particular, the size of a tree is $\bigoh(\abs{V}).$
We do not consider any possible compact representations that might exist for special cases
\footnote{For instance, a path on $n$ vertices can be represented by the integer $n$.}.

We start with the analysis of enumerating the multi-subsets of a set $A$ given by a function $f$ that enumerates its elements.
A pseudo code of an algorithm for such an enumeration is given in Algorithm \ref{alg:MultisetEnumeration}.

\newcommand{\enummultiset}{\textsc{EnumerateMultisets}}

\alglanguage{pseudocode}
\begin{algorithm}[H]
\caption{{\textsc{Enumerating Multisets}}}\label{alg:MultisetEnumeration}
\begin{algorithmic}[1]
\Require{$f$ enumerates a finite set $A=\set{a_1, \ldots}$ of objects.}
\Require{\Call{Bounded}{$f,i$} enumerates the first (at most) $i$ elements of $A$}
\Ensure{Enumerates all the multisets with exactly $k$ elements chosen from the first $upperBound$ elements of $A$.}
\Function{\enummultiset}{$f, k, upperBound$}
\If{$k=0$}
    \State \Yield{$\emptyset$}.
    \State \Return
\EndIf
\State $i \gets 0$
\For{$item \in $ \Call{Bounded}{$f, upperBound$}}
   \State $i \gets i + 1$
   \For{$multiset \in $ \Call{\enummultiset}{$f,k-1,i$}}
       \State \Yield{$multiset \uplus \set{item}$}
   \EndFor
\EndFor
\Return
\EndFunction
\end{algorithmic}
\end{algorithm}

\begin{observation}
The delay time of $\enummultiset$ to emit a multiset $S=\set{s_1, \ldots, s_k}$ is $\bigoh \left( \sum_{j=1}^k D_f(size(s_j)) \right)$ where $D_f(size(s))$ is the delay time of $f$ to emit the object $s$.
\end{observation}
\begin{proof}
By induction on $k$. In order to emit the multiset $S$, $\enummultiset$ makes one recursive invocation of itself and possibly one invocation to ${\textsc{Bounded}(f, upperBound)}$. 
We note that one invocation of ${\textsc{Bounded}(f, upperBound)}$ consists of one invocation of $f$ and an additional bound-check that takes constant time.
By the inductive assumption, the recursive invocation takes $\bigoh \left( \sum_{j=1}^{k-1} D_f(size(s_j)) \right)$ time and 
the invocation of ${\textsc{Bounded}(f, upperBound)}$ to get $item = s_k$ takes $D_f(size(s_k))$ time.
All the other operations including the $\uplus$ operation take constant time.
\end{proof}

Since $size(S) = \left( \sum_{j=1}^k size(s_j) \right)$, we have the following.

\begin{corollary}\label{coro:multisetEnumerationLinear}
If $D_f$ is linear then the output delay time complexity of $\enummultiset$ is linear. 
\end{corollary}

We are now ready to complete the delay time analysis of our enumeration algorithms.
\begin{theorem} \label{thm:outputdelay}
The output delay time complexity of  {\enumrtree} is linear.
\end{theorem}
\begin{proof}
We will show that each of the functions in Algorithm \ref{alg:RootedEnumeration} has linear output delay time complexity,
by induction on the size of the input. 

Clearly, all the functions return in constant time when invoked with $n=1$. 
We now analyze each function separately and assume that the claim inductively holds for all the invocations.
\begin{itemize}
    \item $\enumrtree (n)$ invokes $\enumforest (n-1, n-1)$ that returns a forest $F$ of $n-1$ vertices in $\bigoh(size(F))$ time, by the inductive assumption. 
It then adds a root and its incident edges to return a tree $T$.
Therefore, it takes time $\bigoh(size(T))$ to return a tree $T$.
    
    \item $\enumforest(n,m)$ needs only one invocation of $\enumforest(n,m',\mu)$ to return the forest $F$ returned by the invocation.
    By the inductive assumption, this takes $\bigoh(size(F))$ time.

    \item $\enumforest(n,m.\mu)$ makes one invocation of $\enumforest$ to get $lowForest$. 
    By the inductive assumption, this takes $\bigoh(size(lowForest))$ time. 
    It may possibly invoke $\enummultiset(\enumrtree)$ that takes also linear time, i.e., $\bigoh(size(highForest))$ by the inductive assumption and Corollary \ref{coro:multisetEnumerationLinear}.
    Then it combines the two forests in constant time for an overall time complexity of $\bigoh(size(highForest)) + \bigoh(size(lowForest))$ which is $\bigoh(size(F))$ where $F$ is the returned forest.
\end{itemize}

\end{proof}

\subsubsection{Unranking}
In this part, our goal is to present an algorithm that, given two non-negative integers $n,i$, returns the $i$-th tree in $\RT(n)$ (if such a tree exists).
For this purpose we need the sizes of the sets under consideration.
We now denote 
$rt_n \defined \abs{\RT(n)}$, 
$f_n \defined \abs{\FF(n)}$, 
$f_{n,m} \defined \abs{\FF(n, m)}$,
$f^\leq_{n,m} \defined \abs{\FF^{\leq}(n, m)}$,
$f_{n,m,\mu} \defined \abs{\FF(n, m, \mu)}$, and
$f^\leq_{n,m,\mu} \defined \abs{\FF^\leq(n, m, \mu)}$.
It follows immediately from the definitions that 
\begin{eqnarray}
rt_n & = & f_{n-1} \label{eqn:rec_rt}\\
f_{n,m} & = & \sum_{\mu=1}^{\lfloor n/m \rfloor} f_{n,m,\mu} \label{eqn:rec_f_nm}\\
f^\leq_{n,m} & = & \sum_{m'=1}^m f_{n,m'} \label{eqn:rec_fleq_nm}\\
f^\leq_{n,m,\mu} & = & \sum_{\mu'=1}^\mu f_{n,m,\mu'} \label{eqn:rec_fleq_nmmu}\\
f_n & = & f^\leq_{n,n} \label{eqn:rec_f_n}.
\end{eqnarray}

Finally,
\begin{equation}\label{eqn:rec_f_nmmu}
f_{n,m,\mu}  =  \CC{rt_m}{\mu} \cdot f^\leq_{n-\mu \cdot m, \min \set{n-\mu \cdot m, m-1}}
\end{equation}
where $\CC{i}{j}$ is the number multisets of $j$ elements of a set of $i$ elements, i.e. $\CC{i}{j} = {{i - 1 + j} \choose {j}}$.
Indeed, a forest $F \in \FF_{n,m,\mu}$ is uniquely determined as the union of $\mu$ trees on $m$ vertices and a forest $F'$ of $n - \mu \cdot m$ vertices with trees of at most $\min \set{n - \mu \cdot m, m-1}$ vertices each.

Once this values are computed, we can use the unranking algorithms the pseudo codes of which are given in Algorithm \ref{alg:UnrankingRootedTree}.
The indices of the forests and trees are zero-based and defined by the order that the objects would be returned by the enumeration algorithms (Algorithm \ref{alg:RootedEnumeration}) presented in the previous section. 

\newcommand{\funcrtree}{\textsc{FindRootedTree}}
\newcommand{\funcforest}{\textsc{FindForest}}

\funcrtree$(n,i)$ invokes {\funcforest} to get the $i$-th forest on $n$ vertices and attaches a root to it. 
In order to get the $i$-th forest $F$ on $n$ vertices, {\funcforest} first determines the size $m$ of a largest tree in $F$ (Line \ref{ln:FindForestFindM}), and the index $i'$ of $F$ within the set $\FF_{n,m}$ of these trees (Line \ref{ln:FindForestFindI}).
Then, it determines the number $\mu$ of the largest trees of $F$ (Line \ref{ln:FindForestFindMu})
and the index $i''$ of $F$ in $\FF_{n,m,\mu}$ (Line \ref{ln:FindForestFindIDPrime}).
Finally, with these numbers at hand, it invokes $\funcforest(n.m,\mu,i'')$ to find the forest $F$ within $\FF_{n,m,\mu}$ (Line \ref{ln:FindForestRecurse}).

To get the $i$-th forest $F$ of $\FF_{n,m,\mu}$, we partition $\FF_{n,m,\mu}$ into equivalence classes according to the multiset of the largest trees. 
The size of each equivalence class is $j=f^\leq_{n - \mu \cdot m, \min(n- \mu \cdot m,m-1)}$, i.e. as the number of forests on $n - \mu \cdot m$ vertices with largest tree size of at most $m-1$ (Line \ref{ln:FindForestFindJ}).
We then determine the index of the relevant equivalence class and the index of $F$ in its class (Lines \ref{ln:FindForestFindLFIndex} - \ref{ln:FindForestFindHFIndex}).
With these numbers at hand, it remains to find the multiset of trees that corresponds to the chosen equivalence class (Line \ref{ln:FindForestFindHighForest}), the forest consisting of the remaining trees (Line \ref{ln:FindForestFindLowForest}), and return their union (Line \ref{ln:FindForestReturn}).

\renewcommand{\alg}{\textsc{Unranking Rooted Trees and Forests}}

\alglanguage{pseudocode}
\begin{algorithm}[H]
\caption{{\alg}}\label{alg:UnrankingRootedTree}
\begin{algorithmic}[1]

\Require{$0 \leq i < rt_n$}
\Ensure{Return the $i$-th tree of $\RT(n)$}
\Function{\funcrtree}{$n,i$}
\State $F \gets $ \Call{\funcforest}{$n-1,i$}
\State \Return the tree obtained from $F$ by connecting all its roots to a new root.
\EndFunction
\Statex
\Require{$0 \leq i < f_n$}
\Ensure{Return the $i$-th forest of $\FF(n)$}
\Function{\funcforest}{$n,i$}
\If{$n=0$}
\State \Return the empty forest
\EndIf
\State $m \gets $ the smallest integer $m$ such that $f^\leq_{n,m} > i$ \label{ln:FindForestFindM}
\State $i' \gets i - f^\leq_{n,m-1}$ \label{ln:FindForestFindI}
\State $\mu \gets $ the smallest integer $\mu$ such that $f^\leq_{n,m,\mu} > i'$ \label{ln:FindForestFindMu}
\State $i'' \gets i' - f^\leq_{n,m,\mu-1}$ \label{ln:FindForestFindIDPrime}
\State \Return \Call{\funcforest}{$n,m,\mu,i''$} \label{ln:FindForestRecurse}
\EndFunction
\Statex
\Require{$0 \leq i < f_{n,m,\mu}$}
\Ensure{Return the $i$-th forest of $\FF(n,m,\mu)$}
\Function{\funcforest}{$n,m,\mu,i$}
\State $j \gets f^\leq_{n- \mu \cdot m, \min(n- \mu \cdot m,m-1)}$ \label{ln:FindForestFindJ}
\State $lowForestIndex \gets i \mod j$ \label{ln:FindForestFindLFIndex}
\State $highForestIndex \gets \lfloor i / j \rfloor$ \label{ln:FindForestFindHFIndex}
\LineComment{The choice of the index and the order of forests guarantee}
\LineComment{that the size of the largest tree in $lowForest$ is at most $m-1$}
\State $lowForest \gets$ \Call{\funcforest}{$n - \mu \cdot m,lowForestIndex$} \label{ln:FindForestFindLowForest}
\State $highForest \gets$ the $i$-th multiset in the multi-subsets with $\mu$ elements of $\RT(m)$  \label{ln:FindForestFindHighForest}
\State \Return $highForest \cup lowForest$ \label{ln:FindForestReturn}
\EndFunction

\end{algorithmic}
\end{algorithm}

We now analyze the time complexity of our algorithms.
For this purpose, in the following Lemma, we present a multiset unranking algorithm.

\begin{lemma}\label{lem:multisetComplexity}
Algorithm \ref{alg:MultisetUnranking} returns the $i$-th multiset among the multisets consisting of $k$ elements from the integers $0$ to $n-1$ in $\bigoh(k^2 \log n)$ time. 
\end{lemma}
\begin{proof}
We represent a multiset of $k$ elements by a sequence of length $k$ of its elements in non-increasing order.
We order the multisets in the lexicographic order of their representative sequence. 
Consider the first multiset $S$ in the lexicographic order whose largest element is $\ell$.
We observe that a multiset $S'$ precedes $S$ in the lexicographic order if and only if $S \subseteq [0, \ell-1]$.
Therefore, the zero-based index of $S$ in the lexicographic order is the number of such multisets $S'$, i.e., $\CC{\ell}{k}$
\footnote{As a quick validation,  $\CC{1}{k}={{1-1+k} \choose {k}}=1$ is the index of the first multiset whose largest element is $1$. 
Indeed, there is exactly one set before it, which is the only set whose largest element is zero.}.
Using this observation, the algorithm determines the largest element $\ell$ of the sought set as the largest integer such that $\CC{\ell}{k} \leq i$ (Line \ref{ln:FindMultisetFindL}).
In order to determine the rest of the set, we recursively search for a multiset $\bar{S}$ of $k-1$ elements from the set $[0,\ell]$ (Line \ref{ln:FindMultisetRecurse}).
In order to make the recursive invocation, it remains to determine the index of $\bar{S}$ among the multisets of $k-1$ elements from $[0,\ell]$.
Since the index of the first such set is $\CC{\ell}{k}$, the index of $\bar{S}$ is $i - \CC{\ell}{k}$.

We now proceed with the time complexity.
Since $\CC{\ell}{k}$ can be computed in $\bigoh(k)$ time,
$\ell$ can be determined by a binary search on the $n$ possible values of $\ell$, i.e., in $\bigoh(k \log n)$ time.
Therefore,
\[
T(n,k) = \bigoh( k \log n + T(\ell+1,k-1)) \leq \bigoh(k \log n) + T(n,k-1)
\]
where $T(n,k)$ denotes the running time of finding a multiset with $k$ elements from $[0,n-1]$.
We conclude that $T(n,k)$ is $\bigoh(k^2 ~ \log n)$.
\end{proof}

\newcommand{\funcmultiset}{\textsc{FindMultiset}}

\alglanguage{pseudocode}
\begin{algorithm}[H]
\caption{{\textsc{Unranking Multisets}}}\label{alg:MultisetUnranking}
\begin{algorithmic}[1]
\Ensure{Returns the $i$-th multiset among the multisets consisting of $k$ elements from the integers $0$ to $n-1$}
\Function{\funcmultiset}{$n, k, i$}
\If{$n=1$ or $k=0$}
\State \Return the multiset consisting of $k$ zeros.
\EndIf
\State $\ell \gets$ the largest integer such that $\CC{\ell}{k} \leq i$ \label{ln:FindMultisetFindL}
\State $\bar{S} \gets \Call{\funcmultiset}{\ell+1,k-1,i - \CC{\ell}{k}}$ \label{ln:FindMultisetRecurse}
\State \Return $\set{\ell} \uplus \bar{S}$
\EndFunction
\end{algorithmic}
\end{algorithm}

In the following lemma we analyze the time and space complexity of precomputing all the values \eqref{eqn:rec_rt} thru \eqref{eqn:rec_f_nmmu} needed to run Algorithm \ref{alg:UnrankingRootedTree}.
\begin{lemma}\label{lem:complexity}
For every positive integer $N$, all the values $rt_n, f_n, f_{n,m}, f^\leq_{n,m}$, $f_{n,m,\mu}$, $f^\leq_{n,m,\mu}$ and $\CC{rt_m}{\mu}$ where $n \in [N]$ and $\mu \cdot m \leq n$ can be computed in $\bigoh(N^2 \log N)$ time and space.
\end{lemma}
\begin{proof}
The required memory is dominated by the size of the $f_{n,m,\mu}$ values.
For every $n,m$ pair the required memory is $\bigoh(\frac{n}{m})$.
Therefore, the overall memory requirement is proportional to
\[
\sum_{n=1}^N \sum_{m=1}^{n} \frac{n}{m} = \sum_{n=1}^N n \sum_{m=1}^{n} \frac{1}{m} = \sum_{n=1}^N n \cdot \bigoh(\log n) = \bigoh(N^2 \log N)
\]

We now consider the evaluation of the right hand sides.
\eqref{eqn:rec_rt} and \eqref{eqn:rec_f_n} can clearly be computed in constant time.
Every partial sum, in particular \eqref{eqn:rec_fleq_nm} and \eqref{eqn:rec_fleq_nmmu} can be computed in constant time too.
We note that $\CC{rt_m}{1}=rt_m$ and for $\mu > 1$ $\CC{rt_m}{\mu}$ can be computed from $\CC{rt_m}{\mu-1}$ in constant time (one multiplication, one division and one addition).
Then \eqref{eqn:rec_f_nmmu} can be computed in constant time.
Therefore, the computation time of these values is bounded by the memory size. 
We remain with the computation of $\eqref{eqn:rec_f_nm}$. 
We observe that every $f_{n,m,\mu}$ value is accessed exactly once during the computation of these values.
Therefore, the computation time of $\eqref{eqn:rec_f_nm}$ is bounded by the memory size too.
\end{proof}

We are now ready to analyze the time complexity of the rooted tree unranking algorithms
which assume that the values mentioned in Lemma \ref{lem:complexity} are pre-computed.

\begin{theorem}\label{thm:UnrankingComplexity}
$\funcrtree(n,i)$ returns the $i$-th tree of $\RT(n)$ in $\bigoh(n^3)$ time.
\end{theorem}
\begin{proof}
Let $T(n)$ (resp. $T'(n,m,\mu)$) be the maximum time needed for an invocation of $\funcforest(n,i)$ (resp. $\funcforest(n,m,\mu,i)$) to return.
An invocation $\funcforest(n,i)$ can find the value $m$ in $\bigoh(\log n)$ time by performing a binary search on the table $f^{\leq}_{n,m}$.
Similarly, $\mu$ can be found in $\bigoh(\log n/m)$ time by performing a binary search on all $\lfloor n/m \rfloor$ possible values of $\mu$.
It then invokes $\funcforest(n,m,\mu,i)$ and
the rest of the operations take constant time. 
We therefore have
\begin{equation}\label{eqn:t}
T(n) \leq \bigoh(\log n) + \max_{m, \mu} T'(n,m,\mu).    
\end{equation}

An invocation of $\funcforest(n,m,\mu,i))$, invokes $\funcforest(n - \mu \cdot m, lowForestIndex)$ which takes at most $T(n - \mu \cdot m)$ time.
It finds a multiset of $\mu$ integers from $[0, rt_m]$ which takes $\mu^2 \log rt_m$ time.
Then it invokes $\mu$ times $\funcforest(m-1, i')$ for some $i'$. 
Since the other operations take constant time, we get:
\begin{eqnarray*}
T'(n,m,\mu) & = & T(n - \mu \cdot m) + \mu \cdot T(m-1) + \mu^2 \log rt_m\\
& = & T(n - \mu \cdot m) + \mu \cdot T(m-1) + \bigoh(\mu^2 \cdot m \log m)
\end{eqnarray*}
where the  second equality is due to the fact that the number of rooted trees on $m$ vertices is at most $m!$.
Indeed, a rooted tree on $n$ vertices can be obtained by adding a leaf to any one of the vertices of a rooted tree on $n-1$ vertices, i.e., $rt_m \leq (m-1) rt_m$.

If $T(n)$ is $\bigoh(n)$ then the claim is correct and we are done.
Otherwise, $T(n)$ is $\Omega(n)$.
Therefore, $T(n - \mu \cdot m) + \mu T(m-1) \leq T(n - \mu \cdot m + \mu (m-1)) = T(n - \mu)$ for sufficiently large values of $n$ and $m$.
We proceed  as follows:
\begin{eqnarray*}
T'(n,m,\mu) & = & T(n - \mu \cdot m) + \mu \cdot T(m-1) + \bigoh(\mu^2 \cdot m \log m) \\
& \leq & T(n-\mu) + \bigoh(\mu^2 \cdot m \log m) < T(n - \mu) + \bigoh((\mu \cdot m) (\mu \cdot \log m))\\
& \leq & T(n-1) + \bigoh(n^2).
\end{eqnarray*}
We now substitute in \eqref{eqn:t}
\[
T(n) \leq \bigoh(\log n) + \max_{m, \mu} T'(n,m,\mu) \leq \bigoh(\log n) + T(n-1) + \bigoh(n^2) = T(n-1) + \bigoh(n^2)
\]
concluding that $T(n)$ is $\bigoh(n^3)$.
\end{proof}

\subsection{Free Trees}\label{sec:FreeTrees}
A \emph{free tree} is an unlabeled undirected tree.
In this section, unless stated otherwise, a tree is a free tree.  
Let $\FT(n)$ be the set of free trees on $n$ vertices. 
Fact \ref{fact:centroid} states that a tree is either monocentroidal or bicentroidal.
Moreover, we note that the conditions on the weight stated in Fact \ref{fact:centroid} constitute a dichotomy, thus characterize in which case a tree falls in.
Associate with a monocentroidal tree the rooted tree obtained by designating its centroid as root.
In this case, every subtree of the root has at most $(n-1)/2$ vertices.
With a bicentroidal tree we can associate the two rooted trees obtained by the removal of the edge joining them.
In this case, Fact \ref{fact:centroid} implies that each of them has $n/2$ vertices (thus $n$ is even).
We enumerate all the free trees by first enumerating the monocentroidal ones, 
i.e. those rooted trees on $n$ vertices with subtrees of at most $(n-1)/2$ vertices,
and then all the bicentroidal ones, i.e., the pairs of subtrees of at most $n/2$ vertices (see Algorithm \ref{alg:EnumerateFreeTree}). 
This technique, introduced in \cite{Wilf81}, clearly guarantees that we generate all the (non-isomorphic) free trees.

\renewcommand{\alg}{\textsc{Enumerating Free Trees}}
\newcommand{\enumftree}{\textsc{EnumerateFreeTrees}}

\alglanguage{pseudocode}
\begin{algorithm}[H]
\caption{{\alg}}\label{alg:EnumerateFreeTree}
\begin{algorithmic}[1]

\Ensure{Lazily returns $\FT(n)$}
\Function{\enumftree}{$n$}
\State $\TT \gets \emptyset$
\For{$F \in~$\Call{\funcforest}{$n-1,\lfloor \frac{n-1}{2} \rfloor$}}
    \State $T \gets$ the tree obtained from $F$ by connecting all its roots to a new root.
    \State \Yield $T$
\EndFor
\If {$n \equiv 0 \mod 2$}
    \For{each pair $\set{T_1,T_2} \in {\RT(n/2) \choose 2}$}
    \State $T \gets $ the tree obtained by joining the roots of $T_1$ and $T_2$
    \State \Yield $T$
    \EndFor
\EndIf
\State \Return
\EndFunction
\Statex
\end{algorithmic}
\end{algorithm}

We now proceed with unranking.
Let $t_n = \abs{\FT(n)}$.
Let also, for $i \in \set{1,2}$, $t^{(i)}_n$ be the number of trees on $n$ vertices with $i$ centroids. 
Since a tree has either one or two centroids, we have $t_n = t^{(1)}_n + t^{(2)}_n$.
Moreover,
\begin{eqnarray*}
t^{(1)}_n & = & f^\leq_{n-1,\lfloor (n-1)/2 \rfloor},~\textrm{and}\\
t^{(2)}_n & = & \left\{ \begin{array}{ll}
\CC{rt_{n/2}}{2}      &  \textrm{if~}n \equiv 0 \mod 2\\
0            &  \textrm{otherwise}.
\end{array} \right.    
\end{eqnarray*}

\begin{observation}
$t^{(2)}_n$ can be computed in constant time from $rt_{n/2}$.
\end{observation}
Therefore, all the values $t^{(i)}_n$ can be computed in $\bigoh(N^2 \log N)$ time together with the values in Section \ref{sec:RootedTrees}.

\renewcommand{\alg}{\textsc{Unranking Free Trees}}
\newcommand{\functree}{\textsc{FindFreeTree}}

\alglanguage{pseudocode}
\begin{algorithm}[H]
\caption{{\alg}}\label{alg:FreeTree}
\begin{algorithmic}[1]

\Require{$0 \leq i < t_n$}
\Ensure{Return the $i$-th tree of $\FT(n)$}
\Function{\functree}{$n,i$}
\If{$n \equiv 0 \mod 2$ and $i \geq t^{(1)}_n$}
\State $i \gets i - t^{(1)}_n$
\State $\set{T_1, T_2} \gets$ the $i$-th multiset of two rooted trees from $\RT(n/2)$
\State \Return the tree obtained by connecting the roots of $T_1$ and $T_2$
\Else 
\LineComment{The order of forests and the value of $i$ guarantees}
\LineComment{that the size of the largest subtree is at most $(n-1)/2$}
\State \Return \Call{\funcrtree}{n,i}
\EndIf
\EndFunction
\Statex

\end{algorithmic}
\end{algorithm}

\section{Weighted Trees}\label{sec:WeightedTrees}
In \cite{BielakP12}, the authors show that Fact \ref{fact:centroid} is valid also for weighted trees with strictly positive weights.
Since we will be dealing with weighted trees where we allow zero weights as well, in Section \ref{subsec:StructuralWT}, we extend Fact \ref{fact:centroid} to all non-negatively weighted trees, generalizing the result in \cite{BielakP12}. 
In Theorem \ref{thm:CentroidsNonNegativeWeights}, we show that in this case the number of centroids is not limited to two; 
however they still enjoy some nice properties. 
On the other hand, Fact \ref{fact:centroid} holds for the so-called central centroids. 
Since, due to zero weights, the number of trees of total weight $n$ is unbounded, subsequently, we focus on the so-called canonical weighted trees in which we bound the number of zero-weight vertices in a certain way.
In Corollary \ref{coro:CanonicalCentroids} we show that the number of centroids in canonical weighted trees is again limited, this time by three. In Section \ref{subsec:AlgWeighted}, we explain how centroid-rooted canonical weighted trees are used to generate (non-isomorphic) weighted trees.

\subsection{Structural Properties}\label{subsec:StructuralWT}
\runningtitle{Weighted Trees and Centroids}
A \emph{weighted tree} is a pair $(T,w)$ where $T$ is a tree and $w$ is a non-negative integral weight function on its vertices such that at least one vertex has strictly positive weight.
We extend the weight function $w$ to sets of vertices naturally, so that for every subset $U$ of the vertices of $T$, we have 
$w(U) \defined \sum_{u \in U} w(u)$.
Moreover, we denote $w(T) \defined  w(V(T))>0$.

Let $(T,w)$ be a weighted tree of weight $W$ and $v$ a vertex of $T$. A subtree $T'$ of $v$ in $T$ is \emph{heaviest} if its weight is maximum among all subtrees of $v$ in $T$.
We will term this weight as the \emph{heaviest subtree weight} of $v$ in $(T,w)$ and denote it by $\hs_{T,w}(v)$ or simply $\hs(v)$ when $(T,w)$ is clear from the context. In the degenerated case where $v$ is the only vertex of $T$, i.e., there is no heaviest subtree, $\hs(v)$ is zero.
A \emph{centroid} of $(T,w)$ is a vertex with minimum heaviest subtree weight in $(T,w)$.

The following notions are frequently used in the proof of Theorem \ref{thm:CentroidsNonNegativeWeights}. 
For brevity, we denote a weighted tree $(T,w)$ simply as $T$ and we assume that the weight function is $w$. 
Given a non-trivial path $P$ in a tree $T$ and a vertex $v$ of $P$, denote by $T_{v,P}$ (or $T_v$ if no ambiguity arises) the maximal subtree of $T$ containing $v$, but no other vertex of $P$. 
We denote by $T_P$ the maximal subtree of $T$ that contains all the internal vertices of $P$ but not its endpoints (see Figure \ref{fig:TreePathSeparation}). From these definitions it follows that if $P$ is a non-trivial $(a-b)$-path (i.e., $a \neq b$) then $\set{ V(T_a), V(T_b), V(T_P)}$ is a partition of $V(T)$.
Therefore, $w(T)= w(T_a) + w(T_b) + w(T_P)$.
Note that if $P$ has no internal vertices then $T_P$ is empty.

\begin{figure}
    \centering
        \includegraphics[width=0.7\textwidth]{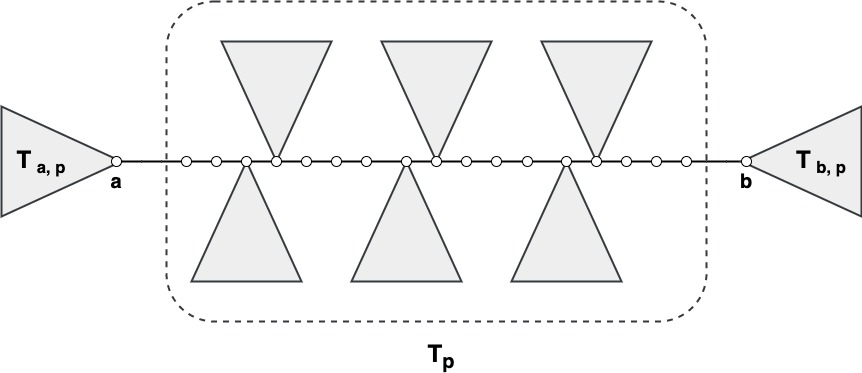}
    \caption{The partition of a tree $T$ implied by a non-trivial $a$-$b$ path $P$}
    \label{fig:TreePathSeparation}
\end{figure}

\begin{theorem}\label{thm:CentroidsNonNegativeWeights}
Let $(T,w)$ be a weighted tree with $w(T)=n$. Then all the centroids of $T$ are located on a path $P$ with endpoints $c_1$ and $c_2$ with $w(T_P)=0$, that is, the weight of the maximal subtree of $T$ containing $P\setminus\{c_1,c_2\}$ is zero. Moreover, if $P$ is non-trivial then $\hs(c_1)=\hs(c_2)=w(T_{c_1,P})=w(T_{c_2,P})=n/2$; otherwise $P$ is a trivial path consisting of one vertex $c$ where $hs(c) < n/2$.
\end{theorem}
\begin{proof}
The claim trivially holds for a tree with one vertex $v$ since $\hs(v)=0<n/2$.
Let $(T,w)$ be a weighted tree on at least two vertices.
We consider two disjoint and complementing cases:
\begin{itemize}
    \item [Case 1:] {$T$ has exactly one centroid $c$:} 
    Then the unique centroid is on a trivial path $P=c$. In this case $T_P$ is empty, thus $w(T_P)=0$. \\
    Since $T$ has at least two vertices, $c$ has at least one heaviest subtree.
    Without loss of generality, let $T_1$ be a heaviest subtree of $c$ and assume by way of contradiction that $w(T_1) = \hs(c) \geq n/2$.
    Then $w(T \setminus T_1) \leq n/2 \leq w(T_1)$.
    Consider the vertex $v$ of $T_1$ that is adjacent to $c$.
    The subtrees of $v$ are $T \setminus T_1$ and other subtrees contained in $T_1 - v$.
    Clearly, for a subtree $T_i$ contained in $T_1 - v$ we have $w(T_i) \leq w(T_1)$.
    Then $\hs(v) \leq w(T_1) = \hs(c)$, contradicting the fact that $c$ is the unique centroid of $T$.

    \item [Case 2] $T$ has at least two centroids:
    Let $c_1, c_2$ be two distinct centroids farthest from each other,
    i.e. there is no other centroid $c_3$ such that $c_1$ (resp. $c_2$) is on the path between $c_2$ (resp. $c_1$) and $c_3$.
    Without loss of generality, let $T_1$ be a heaviest subtree of $c_1$. 
    If $w(T_1)=\hs(c_1)=0$, then all the subtrees of $c_1$ have zero weight.
    Thus $w(c_1) = w(T) = n > 0$ and $c_1$ is the unique centroid of $T$, a contradiction.
    Therefore, $w(T_1)>0$, and in particular $T_1$ is not empty.
    Let $v$ be the vertex of $T_1$ adjacent to $c_1$.
    Consult Figure \ref{fig:Lemma1-1} for this discussion.
    We will now show that $c_2$ is in $T_1$.

    \begin{figure}
    \centering
        \includegraphics[width=0.4\textwidth]{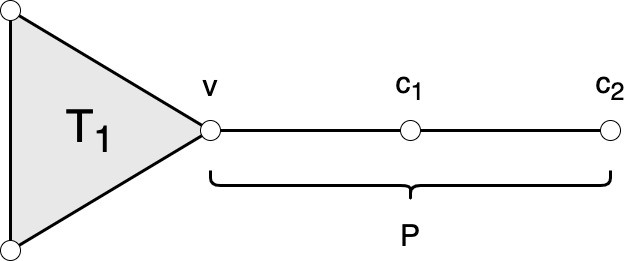}
    \caption{Showing that $c_2$ is in $T_1$}
    \label{fig:Lemma1-1}
    \end{figure}

    Suppose that $c_2$ is not in $T_1$ and let $P$ be the path between $v$ and $c_2$.
    The path $P$ partitions the vertex set of $T$ into subtrees $T_v=T_1, T_{c_2}$ and $T_P$.
    Since $c_2$ is a centroid, we have 
    \[
    w(T_1) = \hs(c_1) = \hs(c_2) \geq w(T_1) + w(T_P) \geq w(T_1).
    \]
    implying that the inequalities hold by equality.
    Then the last equality implies $w(T_P)=0$.
    By the choice of $c_1$ and $c_2$, $v$ is not a centroid.
    Therefore, $\hs(v) > \hs(c_1) = w(T_1)$. 
    All the subtrees of $v$ except $T_P \cup T_{c_2}$ are contained in $T_1$.
    Therefore, $v$ has a unique heaviest subtree, namely $T_P \cup T_{c_2}$. 
    Therefore,
    \[
    w(T_1) < \hs(v)=w(T_P)+w(T_{c_2})=w(T_{c_2}).
    \]
    On the other hand, we have 
    \[
    w(T_1) = \hs(c_1) \geq w(T_{c_2})
    \]
    which contradicts the previous inequality. 
    Therefore, $c_2$ is in $T_1$ which must be the unique subtree of $c_1$.
    By symmetry, $c_1$ is in the unique subtree $T_2$ of $c_1$.
    
    In the rest of the proof we consider the path $P'$ between $c_1$ and $c_2$ that partitions $T$ into subtrees $T_{c_1}, T_{c_2}$ and $T_{P'}$, and we show that the statement holds for $P'$. 
    Since $c_2$ is in the heaviest subtree of $c_1$ we have $\hs(c_1)=w(T_{P'})+w(T_{c_2})$ and symmetrically, $\hs(c_2)=w(T_{P'})+w(T_{c_1})$.
    Since both vertices are centroids, we have $\hs(c_1)=\hs(c_2)$ which implies $w(T_{c_1})=w(T_{c_2})$.
    
    We now show that $w(T_{P'})=0$.
    Suppose that $w(T_{P'})>0$. 
    Then $P'$ contains at least one internal vertex $u$ such that $w(T_{u,P'}) > 0$. 
    Note also that $w(T_{c_1}) > 0$ since otherwise $c_1$ is not a centroid.
    Then
    \begin{eqnarray*}
    \hs(u) & \leq & \max \set{w(T_{c_1})+w(T_{P'})-w(T_u),w(T_u)}\\
    & < & \max \set{w(T_{c_1})+w(T_{P'}), w(T_{c_1})+w(T_u)}\\
    & = & w(T_{c_1}) + \max \set{w(T_{P'}), w(T_u)} 
    = w(T_{c_1}) + w(T_{P'}) = \hs(c_2)
    \end{eqnarray*}
    contradicting the fact that $c_2$ is a centroid.
    Therefore, $w(T_{P'})=0$.
    We conclude that 
    \[
    \hs(c_1)=\hs(c_2)=w(T_{c_1})=w(T_{c_2})=n/2
    \]
    as claimed.
    
    Let $u$ be an internal vertex of $P'$. 
    Then $\hs(u)=\max \set{ w(T_{c_1}),w(T_{c_2} } = n/2$, implying that $u$ is a centroid.
    Therefore all the vertices of $P'$ are centroids.
    We now show that there are no other centroids in $T$.
    There is no centroid in $T_{c_1}$ (resp. $T_{c_2}$) except $c_1$ (resp. $c_2$) by the way $c_1$ and $c_2$ are chosen.
    We remain with the vertices of $T_{P'}$ that are not in $P'$.
    Let $u'$ be a vertex of $T_{P'} \setminus P'$. 
    Then $\hs(u') = w(T_{c_1}) + w(T_{c_2}) = n > n/2 = \hs(c_1)$, thus $u'$ is not a centroid.
\end{itemize}
\end{proof}

The notion of isomorphism extends naturally to weighted graphs in the following way. 
We say that two vertex-weighted graphs (resp. digraphs) are \emph{isomorphic} if there is a bijection between their vertex sets that preserves both adjacencies (resp. directed adjacencies) and weights. 
Our objective is to generate non-isomorphic free weighted trees. 
We will achieve this goal by establishing a one-to-one mapping between the set of free weighted trees and the set of rooted weighted trees in Lemma \ref{lem:iso}.
To this end, we follow the technique employed in \cite{Wilf81} that makes use of the centroids.
Instead, we use the notion of central centroids defined below.

\begin{definition}\label{def:centroid} 
\textbf{(Central centroid:)} A centroid is \emph{central} if it is a center of the path of centroids. 
\end{definition}
Clearly, a weighted tree $T$ has one or two central centroids.
The following definition assumes a predefined total order $\prec$ on the set of weighted rooted trees.
\begin{definition}\label{def:RootedTree} 
\textbf{(Centroid-rooted tree:)} The \emph{centroid-rooted tree} $\dir{T}$ of the weighted tree $T$ is the rooted tree obtained from $T$ by designating a root $r$ as follows: 
    \begin{itemize}
        \item If $T$ has exactly one central centroid $c$ then $r=c$.
        \item Otherwise, $T$ has two adjacent central centroids. Let $T_1$ and $T_2$ be the rooted subtrees of $T$ obtained by the removal of the edge $c_1 c_2$,  where $c_i \in V(T_i)$ and $c_i$ is the root of $T_i$ for $i \in \set{1,2}$. If $T_1$ and $T_2$ are isomorphic, then one of $c_1, c_2$ is arbitrarily designated as $r$. If $T_1$ and $T_2$ are non-isomorphic, then designate $r=c_1$ if $T_1 \prec T_2$ and $r=c_2$ otherwise.
    \end{itemize}
\end{definition}
An illustration of centroid-rooted trees is given in Figure \ref{fig:weighthedblockcuttree} for a weighted block tree, once block trees are introduced in Section \ref{sec:BlockTrees}.
\begin{lemma}\label{lem:iso}
Two weighted trees $T$ and $T'$ are isomorphic if and only if $\dir{T}$ and $\dir{T'}$ are isomorphic.
\end{lemma}
\begin{proof}
($\Leftarrow$) Suppose that $\dir{T}$ and $\dir{T'}$ are isomorphic. Since $T$ and $T'$ are their respective underlying graphs, $T$ and $T'$ are isomorphic.

($\Rightarrow$)
Conversely, suppose that $T$ and $T'$ are isomorphic, and let $f: V(T) \to V(T')$ an adjacency-preserving bijection. 
Then $f$ preserves the non-arcs of $\dir{T}$. 
As for the arcs, they are also preserved except possibly their direction. 
It remains to show that $f$ preserves the directions of the arcs. 
We observe that $f$ preserves the subtrees of a vertex, and in particular their set of weights. Thus $f$ preserves the $\hs$ values, i.e., $\hs_{T'}(f(v))=\hs_T(v)$. Therefore, $f$ preserves the set of vertices attaining the minimum of $\hs$, i.e. the path of centroids. We conclude that $f$ preserves the set of central centroids.
By Theorem \ref{thm:CentroidsNonNegativeWeights}, a weighted tree has either one or two central centroids.
\begin{itemize}
    \item \textbf{Case 1:} If $T$ has exactly one central centroid $c$ then $T'$ has exactly one central centroid $c'$ and $c'=f(c)$. 
    Then the edges of $\dir{T}$ and $\dir{T'}$ are directed towards $c$ and $c'$, respectively. 
    Therefore, $f$ preserves the directions of the arcs.
    
    \item \textbf{Case 2:} If $T$ has two central centroids $c_1$ and $c_2$ then $T'$ has two central centroids $c'_1=f(c_1)$ and $c'_2=f(c_2)$.
    Assume, without loss of generality, that the arc between $c_1$ and $c_2$ is directed from $c_1$ to $c_2$. 
    Let also $T_1, T_2$ (respectively $T'_1, T'_2$) be the two rooted subtrees of $T$ (respectively $T')$ obtained as in Definition \ref{def:RootedTree}. 
    Since the edges of $T_i$ are directed towards $c_i$ and the edges of $T'_i$ are directed towards $c'_i=f(c_i)$, for $i \in \set{1,2}$ $f$ preserves the directions of all the arcs except possibly the arc between $c_1$ and $c_2$, that might be directed from $c'_2$ to $c'_1$. 
    If $T_1$ and $T_2$ are isomorphic that we define a new bijection $f'$ that coincides with $f$ everywhere, except that $f'(c_1)=c'_2$ and $f'(c_2)=c_1$. 
    Then $f'$ is an isomorphism from $\dir{T}$ to $\dir{T'}$. 
    Otherwise, since $c_1 c_2$ is directed from $c_1$ to $c_2$ in $\dir{T}$ we have $T_1 \prec T_2$. 
    Then $T'_1 \prec T'_2$ since $f$ preserves subtrees.
    This contradicts the assumption that $c_1 c_2$ is directed from $c_2$ to $c_1$ in $\dir{T'}$.
\end{itemize}
\end{proof}

It is noted that, since we allow for zero weights, the number of vertices of a weighted tree of a given weight $n$ is unbounded.
Therefore, the following definition is introduced, partitioning the set of weighted trees of a given weight $n$ into a finite set of equivalence classes, by assigning a canonical representative to each equivalence class. The idea is to use these canonical representations to enumerate non-isomorphic weighted tree.

\begin{definition}
\textbf{(Canonical weighted tree:)}
A weighted tree $(T,w)$ is \emph{canonical} if the set of zero weight vertices of $T$ constitutes an independent set and every leaf of $T$ has positive weight.
\end{definition}

Given a weighted tree, we can obtain a unique canonical tree of the same weight by first contracting every subtree consisting of zero-weight vertices to a single vertex and then removing all zero-weight leaves. For weighted trees, we have the following corollary that is based on Theorem \ref{thm:CentroidsNonNegativeWeights}.

\begin{corollary}\label{coro:CanonicalCentroids}
If $(T,w)$ is a canonical weighted tree then exactly one of the following holds.
\begin{enumerate}
    \item $T$ has exactly one centroid $c$ and $\hs(c) < \frac{n}{2}$.
    \item $T$ has exactly two centroids $c_1, c_2$ which are adjacent and the removal of the edge $c_1 c_2$ 
    partitions $T$ into two subtrees $T_1, T_2$ of weight $\frac{n}{2}$ each.
    Moreover, every proper subtree of $T_1$ (resp. $T_2$) has weight at most $\frac{n}{2}-1$. 
    \item $T$ has exactly three centroids with a unique central centroid $c$ with $w(c)=0$ whose removal partitions $T$ into two subtrees of weight $\frac{n}{2}$ each.
\end{enumerate}
\end{corollary}

\begin{proof}
Let $(T,w)$ be a canonical weighted tree and $P$ the path of the centroids of $(T,w)$.
Clearly, the stated cases are mutually-disjoint.

\begin{enumerate}
    \item If $(T,w)$ has exactly one centroid then the first case follows from Theorem \ref{thm:CentroidsNonNegativeWeights}.

    \item If $(T,w)$ has exactly two centroids $c_1, c_2$ then they are clearly adjacent and 
    the removal of the edge $c_1 c_2$ partitions $T$ into two subtrees $T_1, T_2$ each of which has weight equal to $n/2$ by Theorem \ref{thm:CentroidsNonNegativeWeights}. 
    Suppose that $T_1$ has a proper subtree $T'$ of weight $\frac{n}{2}$. Then $w(c_1)=0$ and the vertex of $T'$ adjacent to $c_1$ would be a (third) centroid.    

    \item If $(T,w)$ has at least three centroids, we recall that $w(P)=0$ by Theorem \ref{thm:CentroidsNonNegativeWeights}. 
    In particular, all the internal vertices of $P$ have zero weight. 
    Since no pair of zero-weight vertices is adjacent, $P$ has exactly one internal vertex $v$ and $w(v)=0$. 
    Moreover, since $w(P)=0$, $v$ has no adjacent vertices except the endpoints, say $c_1, c_2$ of $P$. 
    Since $c_1$ and $c_2$ are adjacent to $v$ we have $w(c_1) \neq 0$ and $w(c_2) \neq 0$. 
    By Theorem \ref{thm:CentroidsNonNegativeWeights}, we have $\hs(c_1)=\hs(c_2)=n/2$. 
    Therefore, each one of the subtrees of $v$ has weight of exactly $n/2$, since otherwise either $\hs(c_1)$ or $\hs(c_2)$ would exceed $n/2$, contradicting Theorem \ref{thm:CentroidsNonNegativeWeights}.
\end{enumerate}
\end{proof}

We extend the notion of canonical (weighted) trees to directed weighted trees in a natural way: a directed weighted tree is canonical if its underlying weighted tree is canonical. Now, we use the following consequence of Lemma \ref{lem:iso} and Corollary  \ref{coro:CanonicalCentroids} to enumerate non-isomorphic weighted trees. 

\begin{corollary}\label{cor:iso}
Two canonical weighted trees $T$ and $T'$ are isomorphic if and only if $\dir{T}$ and $\dir{T'}$ are isomorphic and both are canonical.
\end{corollary}

Therefore, the generation of canonical weighted trees is equivalent to the generation of canonical centroid-rooted trees. In the rest of this section, we provide the implementation details of such a generation.
\subsection{Generation algorithms for weighted trees}\label{subsec:AlgWeighted}
We now describe the modifications to the algorithms from Section \ref{sec:Trees} needed to generate weighted trees. 
For this purpose, we first consider positive-weighted trees and then proceed with weighted trees.

\runningtitle{Positive-weighted trees}
In order to extend the rooted tree generation algorithms to weighted trees (with no zero weights allowed), 
it is sufficient to modify the parts of the algorithms that generate vertices, i.e., the root of a tree.
Specifically, we have to modify the first functions of both Algorithm \ref{alg:RootedEnumeration} and Algorithm \ref{alg:UnrankingRootedTree} in 
such a way that, instead of connecting a root to a forest of $n-1$ vertices,
they consider each of the possible weights $w \in [n]$ for the root and connect it to a forest of weight $n-w$.

As for rooted trees, since Fact \ref{fact:centroid} holds also for positive-weighted trees, again, it is sufficient to modify Algorithm \ref{alg:EnumerateFreeTree} as follows: when generating the root of a mono-centroidal tree, 
it tries all the $w \in [n]$ and then generates a forest of weight $n-w$ with maximum subtree-weight of at most $n/2-1$. 

As our goal is to derive a single generic algorithm that can generate both unweighted and positive-weighted trees, we can generalize as follows. 
Define two parameters, say $minw$ and $maxw$
and restrict the possible weights to $w \in [minw..maxw]$ instead of $w \in [n]$.
Clearly, setting $minw=maxw=1$ corresponds to unweighted trees, and $minw=1$, $max_w=n$ corresponds to positive weighted trees.

\runningtitle{Canonical weighted trees}
Extending the above approach to weighted trees (that allow zero weights) seems, at first glance, to be as simple as setting $minw=0$.
However, zero-weight vertices introduce two problems: 
a) enabling the possibility of three centroids, and b) disabling the possibility of adjacent zero-weight vertices.
The solution of the first problem can be easily derived from Corollary \ref{coro:CanonicalCentroids}.
In this part we will address the second problem.
First, consider rooted trees. 
If a root vertex has zero weight then none of its children may have zero weight.
A possible solution is to divide vertices into two sets, say, "green" and "white" vertices
where "green" stands for "positive-weighted" and "white" stands for "zero-weight".
The children of white vertices should be all green, and the children of green vertices may be either green or white.
To generalize this approach, for each color $c$ define three parameters.
Namely, the minimum and maximum vertex weights $minw(c), maxw(c)$, and a set $chld(c)$ of possible colors for its children.
In the sequel, the color of a rooted tree is the color of its root,
and the color-set $cs(F)$ of a forest $F$ is the set of colors of its trees.
In order to generate a tree of a certain color $c$, one has to try all the possible weights $w \in [minw(c)..maxw(c)]$,
then one has to generate all the forests $F$ of weight $n-w$ with $cs(F) \subseteq cs_c$. 
Table \ref{tbl:colors1} shows the setting for weighted and unweighted trees.
The set of weighted trees is the union of green and white trees whose parameters are defined as mentioned above.
The set of unweighted trees is exactly the set of gray trees.
Indeed, the root $r$ of a gray tree is gray and thus $w(r) \in [1..1]$, i.e., $w(r)=1$. 
Moreover, every subtree of $r$ is gray. 
This inductively defines the set of unweighted trees.
\begin{table}[h]
\centering
\begin{tabular}{|l|l|l|l|}
\hline
 $c$  & $minw(c)$ &  $maxw(c)$  &   $chld(c)$   \\ 
\hline
Gray  & 1         & 1           & \{ Gray \}     \\ \cline{1-1}
Green & 1         & $\infty$    & \{White, Green\}  \\ \cline{1-1}
White & 0         & 0           & \{ Green \}    \\ \cline{1-1}
\hline
\end{tabular}
\caption{The color parameter values used to construct weighted and unweighted trees.}\label{tbl:colors1}
\end{table}

\section{Weighted Block Trees}\label{sec:BlockTrees}
In this section, we extend the previous notions and algorithms to weighted block trees. This allows us to enumerate and generate u.a.r. block graphs which form a subclass of chordal graphs where each block consists of a complete graph.
\subsection{Structural Properties}
A \emph{block} of a graph is a maximal biconnected component of it.
The \emph{block tree} (a.k.a. \emph{block-cutpoint} tree \cite{harary1971graph}, a.k.a \emph{super graph} \cite{E73}) of a connected graph $G$ 
is a tree each vertex of which is associated either with a cut vertex of $G$ or with a block of $G$.
We will refer to a vertex associated with a block as a \emph{block vertex} and name it after the block it is associated with.
With a slight abuse of terminology, we refer as a cut vertex to a vertex (of the block tree) associated with a cut vertex of $G$,
and name it after its associated cut vertex (in $G$).
A cut vertex $v$ is adjacent (in the block tree) to a block vertex $B$ if and only if $v$ is a vertex of $B$ (in $G$).
Cut vertices (resp. block vertices) are pairwise non-adjacent in a block-tree.

It follows from the above definitions that the degree of a cut vertex (in the block tree) is the number of blocks to which it belongs (in the graph).
Since a cut-vertex belongs to at least two blocks it cannot be (associated with) a leaf (of the block tree).
Stated differently, all the leaves of a block tree are associated with blocks. 
It also follows from the definitions that block vertices and cut-vertices constitute the (unique) bipartition of the block tree.
Therefore, the leaves of the block tree are in one part of the bipartition.
In other words, in any bi-coloring of the block tree, the set of leaves is monochromatic.
Equivalently, the distance between any pair of leaves is even.
In fact this is one characterization of block trees \cite{harary1971graph}.

An example of the block cut-vertex tree and their centroids is given in Figure \ref{fig:weighthedblockcuttree}. 

\begin{figure}[H]
    \centering
        \includegraphics[width=0.7\textwidth]{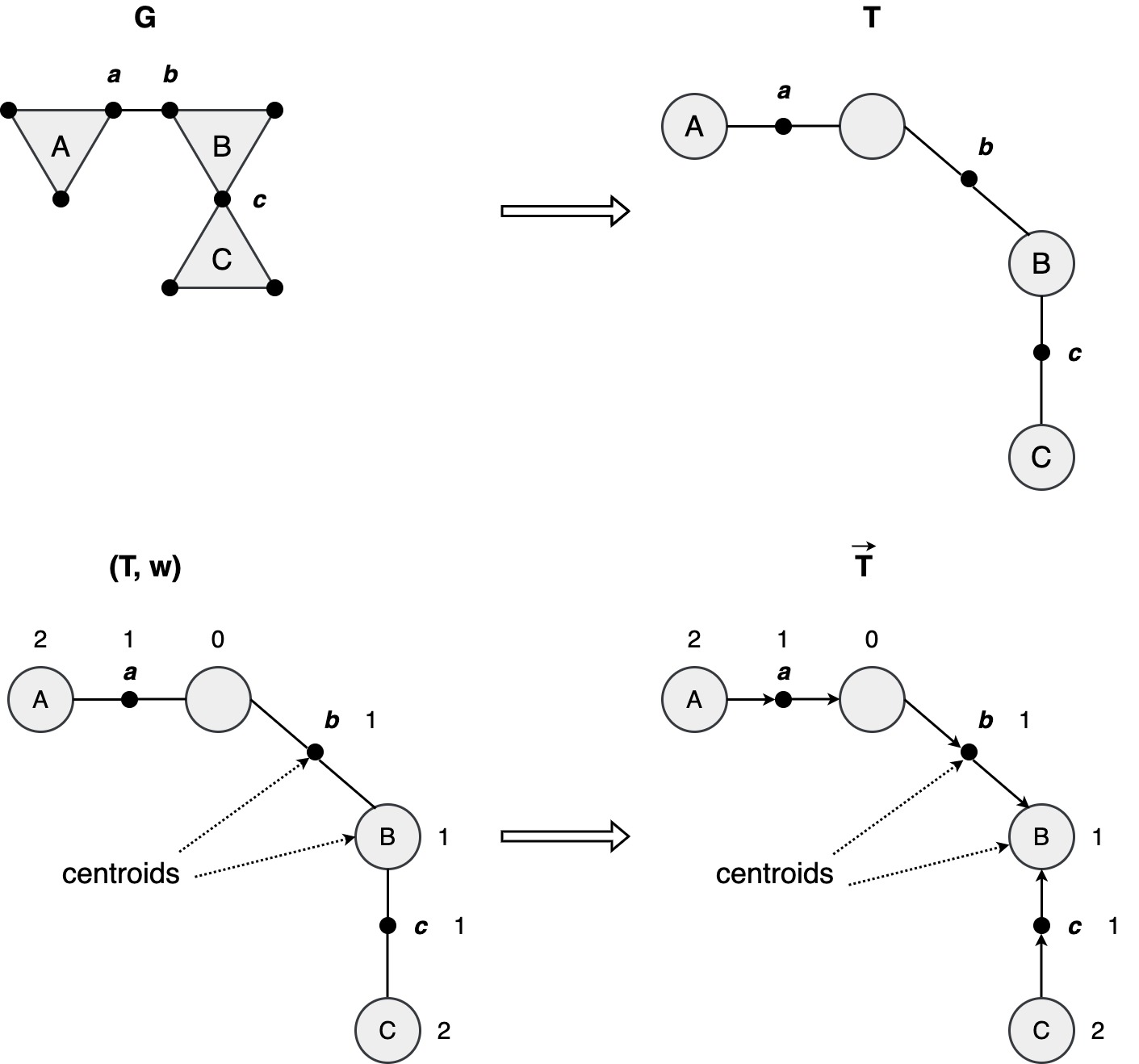}
    \caption{
    A graph $G$, its block tree $T$, its weighted block tree $(T,w)$ and its centroid-rooted tree.
    The graph $G$ depicted at the top-left has three cut-vertices, namely $a, b$ and $c$. 
    It has four blocks which are the three triangles $A, B, C$ and the $K_2$ induced by $\set{a,b}$. 
    Every cut vertex of $G$ is in exactly two blocks and every such block contains at most two cut-vertices. 
    Thus we obtain the block cut-vertex tree $T$ which is the path depicted at the top-right. 
    The block $\set{a, b}$ contains only cut-vertices, the triangle $B$ contains two cut-vertices (namely $b$ and $c$) and each of the other two triangles contains one cut-vertex each. 
    Therefore, the number of remaining vertices in these blocks are 0, 1, 2 and 2, respectively. 
    These weights are shown at the bottom-right together with the weights of the cut-vertices which are all one. 
    The sum of the weights is therefore $0+1+2+2+3\cdot 1=8$ which is equal to the number of vertices of $G$. 
    The centroids of $(T,w)$ are $b$ and $B$. 
    The removal of the edge between them separates $(T,w)$ into two subtrees of weight 4, each. 
    At the bottom-right is depicted the centroid-rooted tree in which all the edges are directed towards these centroids. As our example lacks the definition a total order $\prec$, the orientation of the edge between the two centroids is chosen arbitrarily. 
    }
    \label{fig:weighthedblockcuttree}
\end{figure}

\begin{definition}
\textbf{(Weighted block tree:)} A \emph{weighted block tree} is a weighted tree $T$ with the following weight function.
The weight of a cut vertex is 1 and the weight of a block vertex is the number of (possibly zero) vertices in the block that are not cut vertices.
\end{definition}
It follows from the definition that if $T$ is the block tree of a graph $G$, then every vertex of $G$ contributes exactly 1 to $w(T)$, i.e., $w(T)=\abs{V(G)}$.
This is because blocks intersect only at cut vertices.
We now note that a block tree is canonical: a) zero weight vertices, being blocks are pairwise non-adjacent, b) the weight of a leaf cannot be zero since a block contains at least two vertices and a leaf block contains at most one cut-vertex.

We now restate Corollary \ref{coro:CanonicalCentroids} for the special case of block trees.

\begin{corollary}\label{cor:BlockCentroids}
If $T$ is the weighted block tree of a graph $G$ on $n$ vertices exactly one of the following holds.
\begin{enumerate}
    \item $T$ has exactly one centroid $c$ and $\hs(c) < \frac{n}{2}$.
    \item $T$ has exactly two centroids $v, B$ where $v \in B$ is a cut vertex of $G$ and the removal of the edge $v B$ 
partitions $T$ into two subtrees $T_1, T_2$ of weight $\frac{n}{2}$ each.
Moreover, every proper subtree of $T_1$ (resp. $T_2$) has weight at most $\frac{n}{2}-1$.
    \item $T$ contains a zero-weight (block) vertex $B$ which is the unique central centroid of $T$ whose removal partitions $T$ into two subtrees of weight $\frac{n}{2}$ each. Therefore, the block $B$ is a $K_2$.
\end{enumerate}
\end{corollary}

Let $f: V(G) \to V(G')$ be an isomorphism.
Since $f$ preserves adjacencies, it preserves paths.
Therefore, it preserves cut-vertices and connected components.
On the other hand, $f$ preserves induced subgraphs, i.e., for any subset $U$ of $V(G)$, $G'[f(U)]$ is isomorphic to $G[U]$.
Being maximal subgraphs that do not contain cut-vertices, $f$ preserves blocks.
It also preserves the memberships of cut-vertices in blocks.
These facts imply an isomorphism from the block tree of $G$ to the block tree of $G'$.
Since the sizes of the blocks and the number of cut-vertices in a block are preserved
we get the following observation.
\begin{observation}
\label{obs:blockiso}
If two graphs $G$ and $G'$ are isomorphic then their weighted block trees are isomorphic.
\end{observation}

\begin{definition}
Let $T$ be the block tree of a graph $G$ and $T'$ a subtree of $T$. 
We denote by $G[T']$ the subgraph of $G$ induced by the union of all blocks in $T'$. 
Let $T_1, \ldots,  T_k$ be the subtrees of a (cut or block) vertex of $T$. 
We term the subgraphs $G[T_1], \ldots , G[T_k]$ as the \emph{subgraphs} defined by this vertex.
\end{definition}

\runningtitle{Block Graphs}
The \textit{block graph} $B(G)$ of a given graph $G$ is the intersection graph of its blocks. 
A graph is a block graph if it is isomorphic to $B(G)$ for some graph $G$.
\begin{theorem}\label{thm:block}
There is a one-to-one correspondence between the class of connected block graphs on $n$ vertices and the class of weighted block-trees of weight $n$.
\end{theorem}
\begin{proof}
Let $\bcw(G)$ be the weighted block tree of a connected block graph $G$.
We have to show that this mapping is a bijection, i.e.,
(a) for every weighted block tree $(T,w)$, there exists a block graph $G$ on $w(T)$ vertices such that $\bcw(G)=(T,w)$ and
(b) if $(T,w)=\bcw(G)=\bcw(G')$ for two block graphs $G$ and $G'$ then these graphs are isomorphic.

It is well-known that a graph is a block graph if and only if its blocks are cliques \cite{Harary63}.

(a) Let $(T,w)$ be a weighted block tree. 
We construct a block graph $G$ as follows.
For every vertex $B$ of $T$ that must correspond to a block (i.e., one that is at even distance from the leaves) construct a clique on $w(B)+d_T(B)$ vertices.
For every other vertex $v$ of $T$ (that must correspond to a cut-vertex) choose a vertex (that is yet not chosen) from each one of the $d_T(v)$ blocks adjacent to $v$ in $T$ and identify them.
Note that the number of vertices created initially is $\sum_B (w(B)+d_T(B)) = \sum_B w(B) + \sum_B d_T(B) = \sum_B w(B) + \sum_v d_T(v)$ where the last equality holds because $(B,v)$ is a bipartition of $T$.
At every identification operation $d_T(v)$ vertices are replaced by a single vertex, i.e., the number of vertices drops by $d_T(v)-1$.
Therefore, the number of vertices of $G$ is $\sum_B w(B) + \sum_v d_T(v) - \sum_v (d_T(v) - 1) = \sum_B w(B) + \sum_v 1 = w(T)$.

(b) Let $B_1, \ldots, B_k$ be the blocks of $G$, and let $B'_1, \ldots, B'_k$ be the blocks of $G'$ corresponding to the same block vertices in $T$ in the same order.
Similarly, let $v_1, \ldots, v_\ell$ be the cut-vertices of $G$, and let $v'_1, \ldots, v'_\ell$ be the cut-vertices of $G'$ corresponding to the same vertices in $T$ in the same order.
The blocks $B_i$ and $B'_i$ have the same number of vertices (i.e., $w(B_i)+d_T(B_i)$), for every $i \in [k]$.
Therefore, $B_i$ and $B'_i$ are cliques of the same size.
Let $f:V(G) \to V(G')$ such that $f(v_j)=v'_j$ for every $j \in [\ell]$,
and for every $i \in [k]$, the $w(B_i)=w(B'_i)$ non-cut-vertices of $B_i$ are mapped to the non-cut-vertices of $B'_i$ arbitrarily. 
The mapping $f'$ is clearly a bijection.
We note that $N_G(v) = \bigcup_{v \ni B_i} B_i$, i.e., the union of the blocks that contain $v$.
For a non-cut-vertex, this neighborhood is clearly preserved by $f'$.
Let $v_j$ be a cut-vertex of $G$ that is contained in blocks $B_{j_1}, \ldots$.
Then $v_j$ is adjacent to $B_{j_1}, \ldots$ in $T$.
Since $v'_j$ is mapped to the same vertex of $T$, it is contained in blocks $B'_{j_1}, \ldots$.
Therefore, the neighbourhood of $v_i$ is preserved by $f'$.
We conclude that $f'$ is an isomorphism.
\end{proof}

\subsection{Generation algorithms for weighted block trees}
We now describe the modifications to the algorithms in the previous sections needed to generate weighted block trees, thus connected block graphs by Theorem \ref{thm:block}. 
We will call the resulting generation algorithm {\bgg}.

A \emph{rooted block tree} is a block tree with an arbitrarily chosen root.
Color the vertices of a rooted block tree with two colors "yellow" and "red" 
where cut vertices are yellow and block vertices are red.

The structural properties of the block trees can be stated as follows.
\begin{itemize}
\item The children of a yellow vertex are all red. 
\item The children of a red vertex are all yellow.
\item The weight of a yellow vertex is 1.
\item The weight of a red vertex is unrestricted. However,
\item The weight of a red vertex that is a leaf is at least 1.
\item A yellow vertex cannot be a leaf.
\end{itemize}

The first four conditions above can be expressed by the framework described in Section \ref{subsec:AlgWeighted},
In order to express the last two conditions we first restate them as follows:
\begin{itemize}
\item The weight of a red tree is at least 1.
\item The weight of a yellow tree is at least 2.
\end{itemize}
Now it is sufficient to add a new parameter $mintw(c)$ which is the minimum weight of a subtree whose root is colored $c$.
We now note that a simple modification of the algorithms is sufficient to support this parameter.
Namely, when considering the weight of $m$ of the largest subtree we do not allow values below $mintw(c)$.
Table \ref{tbl:colors2} extends Table \ref{tbl:colors1} with the new parameter $minwt$ and colors introduced in this section.

\begin{table}[h]
\centering
\begin{tabular}{|l|l|l|l|l|}
\hline
 $c$  & $minw(c)$ & $maxw(c)$ & $mintw(c)$ & $chld(c)$   \\ 
\hline
Gray  & 1         & 1         &  1         & \{ Gray \}   \\ \cline{1-1}
Green & 1         & $\infty$  &  1         & \{ White, Green\}   \\ \cline{1-1}
White & 0         & 0         &  1         & \{ Green \}   \\ \cline{1-1}
Yellow & 1        & 1         &  2         & \{ Red \}     \\ \cline{1-1}
Red    & 0        & $\infty$  &  1         & \{ Yellow \}  \\ \cline{1-1}
\hline
\end{tabular}
\caption{The color parameter values used to construct weighted, unweighted and block trees.}\label{tbl:colors2}
\end{table}

\section{A General Framework and Implementation}\label{sec:framework}
\runningtitle{Framework}
In this section we formally describe the framework developed in the previous sections.
For simplicity, we confine ourselves to the case where $\abs{chld(c)}=1$ that is sufficient to generate unweighted, positive-weighted and weighted block trees.
We extend the notation given in Section \ref{sec:RootedTrees} by 
a) replacing the number of vertices $n$ by the weight $w$ of a tree, and 
b) introducing the color parameter $c$.
For instance, $\RT_{w,c}$ is the set of all trees of total weight $w$ colored $c$.
The correctness of the following equations that extend equations \eqref{eqn:rec_rt} - \eqref{eqn:rec_f_nmmu} can be easily verified.

\begin{eqnarray}
rt^\leq_{w,c,m}  & = & 
\left\{
\begin{array}{cc}
0     &  \textrm{if~} w < mintw(c)\\
\sum_{r=minw(c)}^{\min(maxw(c),w)} f^\leq_{w-r,chld(c),\min(m,w-r)}     & \textrm{otherwise}
\end{array}
\right.
\label{eqn:rec11}\\
rt_{w,c}  & = &  rt^\leq_{w,c,w} \label{eqn:rec12}\\
f_{w,c,m,\mu} & = & \CC{rt_{m,c}}{\mu} \cdot f^\leq_{w-\mu \cdot m, c, \min \set{w-\mu \cdot m, m-1}} \label{eqn:rec13}\\
f^\leq_{w,c,m,\mu} & = & \sum_{\mu'=1}^\mu f_{w,c,m,\mu'} \label{eqn:rec14}\\
f_{w,c,m} & = & f^\leq_{w,c,m,\lfloor \frac{w}{n} \rfloor}\label{eqn:rec15}\\
f^\leq_{w,c,m} & = & 
\left\{ \begin{array}{ll}
1    & \textrm{if~} w=0  \\
\sum_{m'=mintw(c)}^{m} f^\leq_{w,c,m',\lfloor w / m' \rfloor}   & \textrm{otherwise}
\end{array}
\right. \label{eqn:rec16}\\
f_{w,c} & = & f^\leq_{w,c,w} \label{eqn:rec17}
\end{eqnarray}

It can be easily verified that these values can be computed in $\bigoh(N^2 \log N)$ time.
Moreover, using the technique used for unweighted trees one can show that the output delay time complexity of the enumeration algorithm is linear.
Finally, the complexity analysis of unranking can be computed along the same lines of Theorem \ref{thm:UnrankingComplexity}.

\runningtitle{Implementation}
The algorithms were implemented in Python.
The source code of the library can be found in \cite{PythonImpl}.
In addition, a Google Colaboratory (or Colab, for short) demonstrating the usage of the library is published in \cite{Colab} for convenience.
Figure \ref{fig:CAT} shows the running times taken from a sample run of the Colab. 
The chart shows the average (amortized) construction time of a rooted tree of weight $n$ when all the trees of weight $n$ are constructed. 
We enumerated all (12,826,228) rooted trees on 20 vertices using a standard (free) Colab with 12GB of memory.
It can be seen that except a few irregularities when the number of trees is small, 
that are apparently caused by system overhead in various levels, the construction time per tree remains constant.
This is an experimental evidence that our algorithms are CAT.

\begin{figure}[H]
    \centering
    \includegraphics[width=0.5\textwidth]{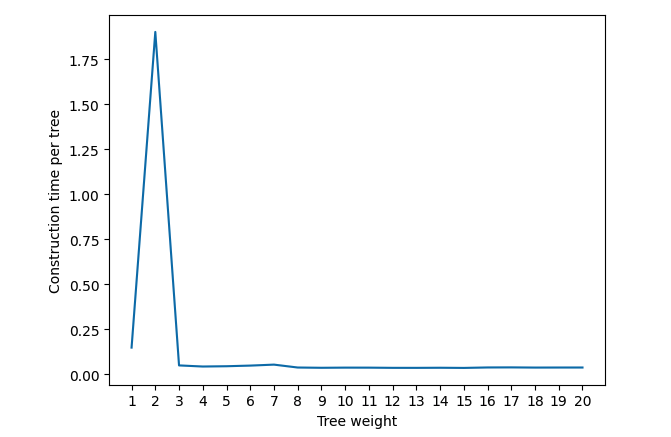}
    \caption{Amortized enumeration time in milliseconds.}
    \label{fig:CAT}
\end{figure}

See Figure \ref{fig:WeightedTreesFive} in the Appendix for a list of positive-weighted trees of weight 5, and 
Figure \ref{fig:BlockTreesSix} in the Appendix for a list of all weighted block trees of graphs on 6 vertices categorized by their types. 
See also Table \ref{tab:NumberOfBlockTrees} in the Appendix for the numbers of (non-isomorphic) block trees of total weight $n$ for $n \in [30]$, or equivalently, the numbers of (non-isomorphic) connected block graphs with $n$ vertices for $n\in [30]$. 
These graphs up to 19 vertices are enumerated using {\bgg} and can be downloaded from \cite{HoGBlock}. 
Block graphs with more vertices can be generated using our open source codes in \cite{PythonImpl} and \cite{Colab}, which can also be easily adapted for a generation u.a.r..

\section{Conclusion}\label{sec:conclusion}
We presented a generic framework (as formalized in Section \ref{sec:framework}) to solve enumeration problems of structured graphs using some tree structure such as tree decomposition, clique-width decomposition, block-trees, etc. 
We demonstrate this approach for block graphs using their block trees that uniquely identifies them. In the sequel we mention possible future works.

We strongly believe that our enumeration algorithms are CAT. Namely, they have Constant Amortized Time. This is due to the fact that the order of the generated trees is similar to the order of trees in previous algorithms that are proven to be CAT.
Specifically, most of the generated trees are obtained by a few changes to the preceding tree in the order. In this work, we provided experimental evidence of this. 
    
We have shown that the output delay time of our enumeration algorithms is linear. Whereas, as said above, most of the trees are obtained by a few changes to the preceding tree, this does not hold for all the trees. Some trees are very different than their preceding tree. Therefore, they cannot be obtained from the previous tree in the order, in constant time. 

Enumerating weighted/block trees in some other order that will allow for sub-linear delay time and the rigorous analysis of the amortized enumeration time are left as open questions. 

\runningtitle{Future works}

    U.a.r. generation:
    As mentioned in the introduction, one can generate objects uniformly at random using a counting and an unranking algorithm.
    However, this approach is not time-efficient since at each step of the generation of the object, one has to compute the exact path to get an object of the given rank.
    A more time-efficient approach would be to make random decisions at each step
    which can be done by simple computations on our coefficients.

    Improving the time complexity of unranking:
    In Lemma \ref{lem:multisetComplexity}, we have shown that the $i$-th multiset of $k$ elements over a set of $n$ elements can be found in $\bigoh(k^2 \log n)$ time. 
    In the proof, we used a subroutine that performs a binary search on an array of exponential size each element of which is a binomial coefficient that can be computed in $\bigoh(k)$ time.
    This can be possibly improved by performing a "fuzzy" binary search using an oracle that computes approximations of binary coefficients, faster.

    Improving the complexity analysis of unranking:
    The analysis in Theorem \ref{thm:UnrankingComplexity} might not be tight. 
    It might be possible to show that the time-complexity of unranking is $\bigoh(n^2 \log n)$.

\runningtitle{Extensions}
We note that our framework can be applied to other types of tree-based decompositions of graphs whenever there exists a bijection between the graphs and their decompositions.

In some cases such a bijection does not exist as in the case of cactus graphs that are exactly the class of graphs whose blocks are either cycles or $K_2$ s.
In this case, one weighted block tree may correspond to multiple cactus graphs.
As opposed to block graphs, this requires a more complex algorithm since, unlike the cliques in which every bijection is an isomorphism, cycles induce a smaller symmetry group. 

In other cases where the blocks do not exhibit symmetry properties that can be exploited, one may resort to isomorphism checks of the individual blocks. This approach can be promising especially when the blocks have bounded size. 

Our approach can be promising in generating several subclasses of chordal graphs, such as chordal graphs of bounded tree-width and $k$-separator chordal graphs.
In these cases, one can first generate a weighted clique-tree and then generate the corresponding chordal graphs and perform isomorphism checks of the subgraphs.

\bibliographystyle{abbrv}
\bibliography{GraphTheory,References}

\newpage
\section*{Appendix}
\appendix
\begin{figure}[h]
    \centering
    \begin{subfigure}{\textwidth}
        \centering
        \includegraphics[width=0.22\textwidth]{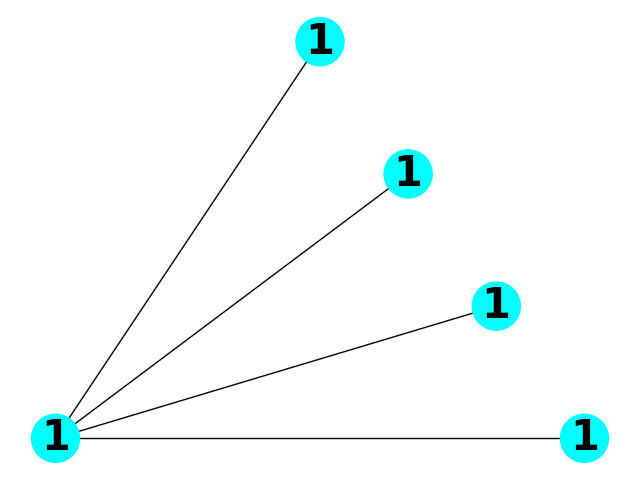}\qquad        \includegraphics[width=0.22\textwidth]{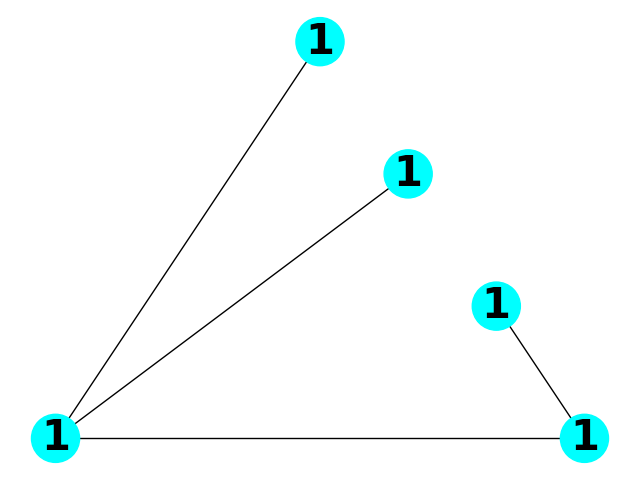}\qquad
        \includegraphics[width=0.22\textwidth]{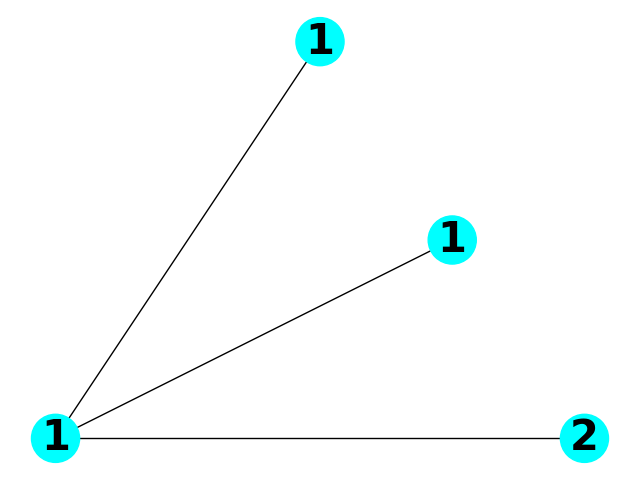}\qquad
        \includegraphics[width=0.22\textwidth]{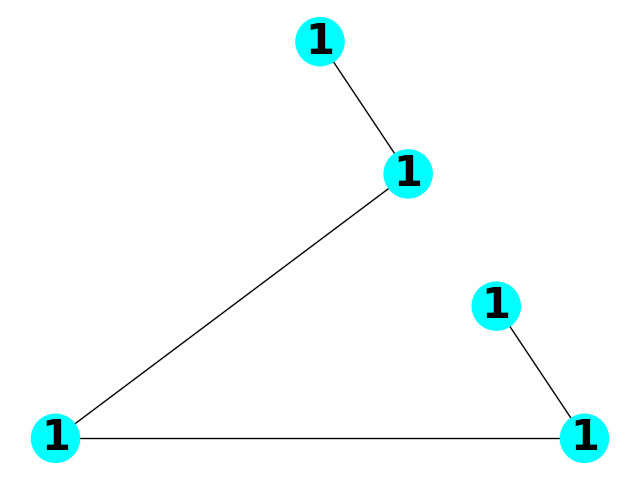}\qquad
        \includegraphics[width=0.22\textwidth]{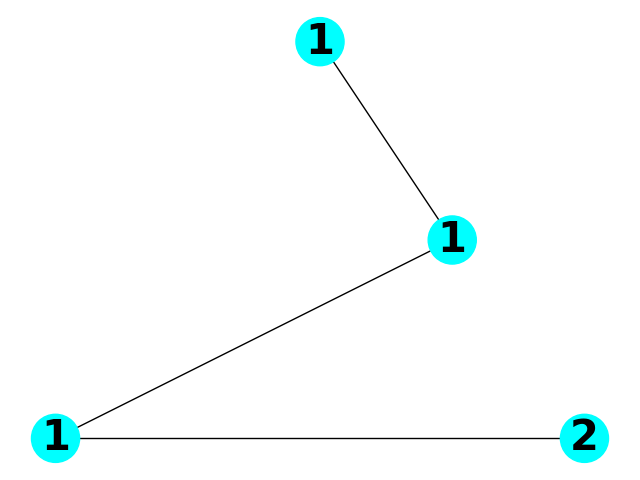}\qquad
        \includegraphics[width=0.22\textwidth]{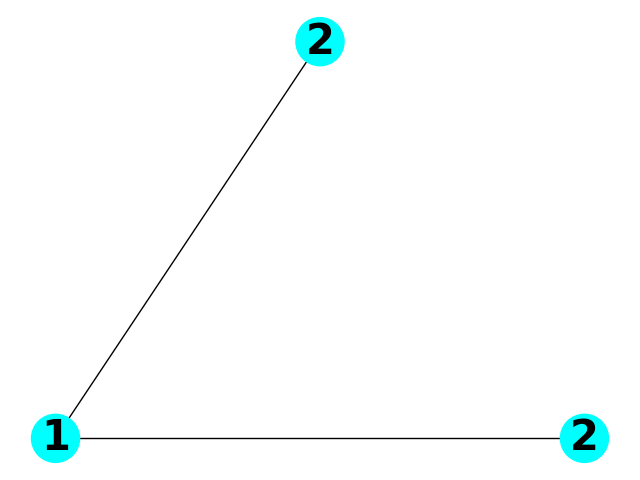}\qquad
        \includegraphics[width=0.22\textwidth]{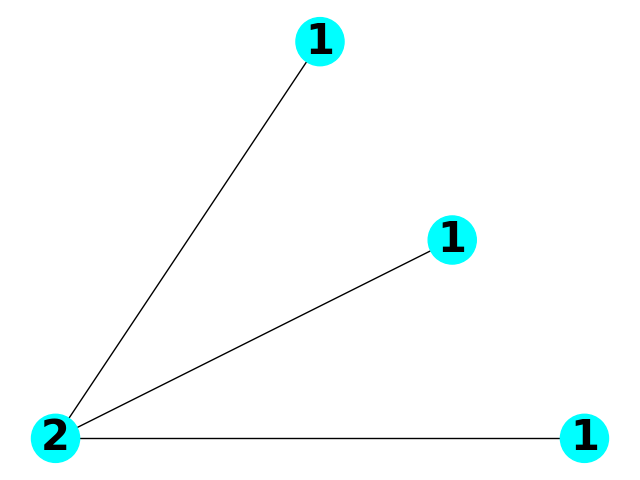}\qquad
        \includegraphics[width=0.22\textwidth]{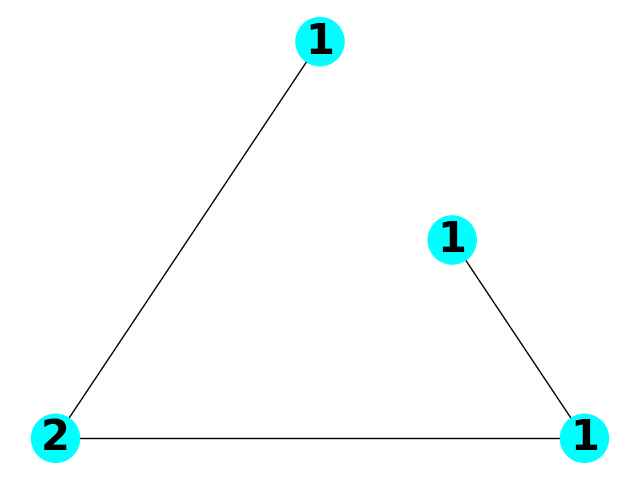}\qquad
        \includegraphics[width=0.22\textwidth]{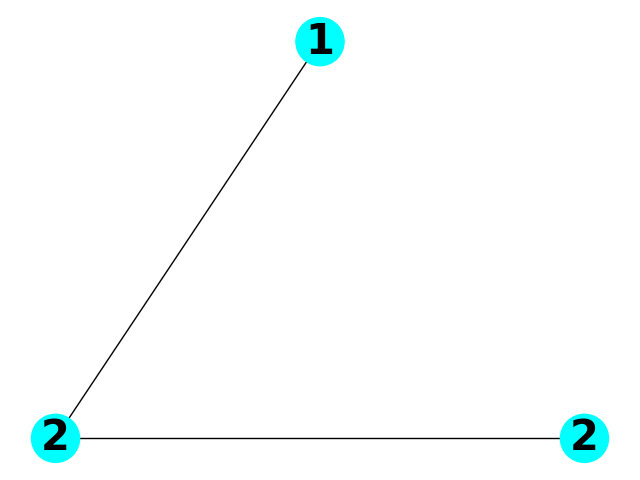}\qquad
        \includegraphics[width=0.22\textwidth]{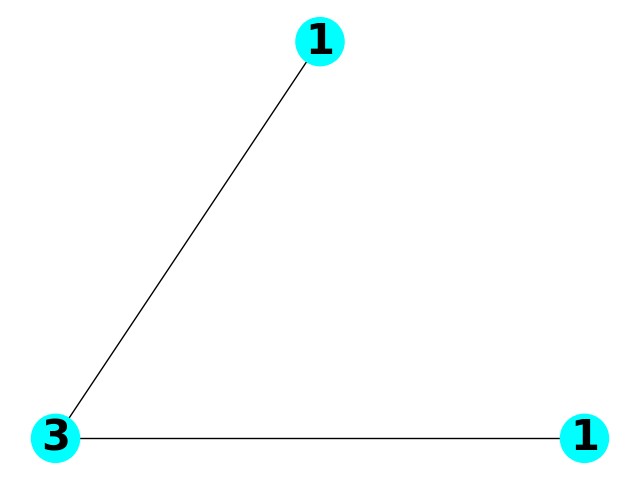}\qquad
        \includegraphics[width=0.22\textwidth]{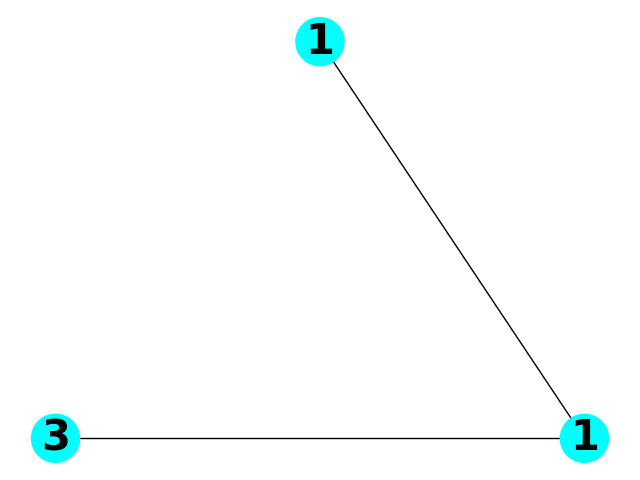}\qquad
        \includegraphics[width=0.22\textwidth]{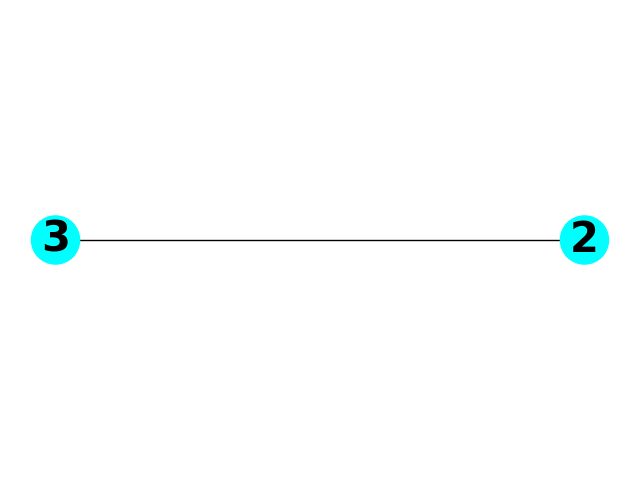}\qquad
        \includegraphics[width=0.22\textwidth]{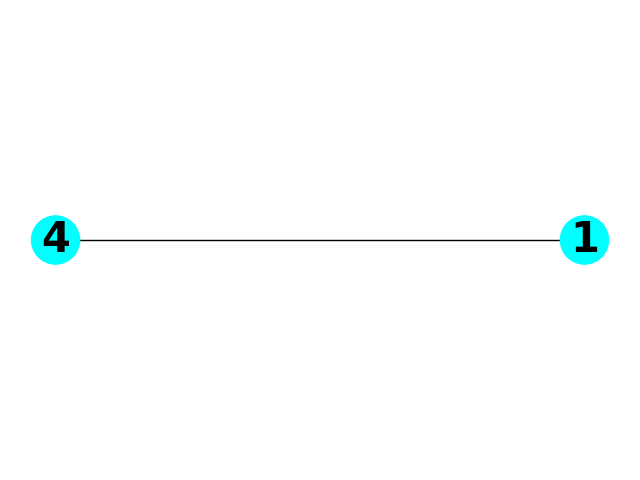}\qquad
        \includegraphics[width=0.22\textwidth]{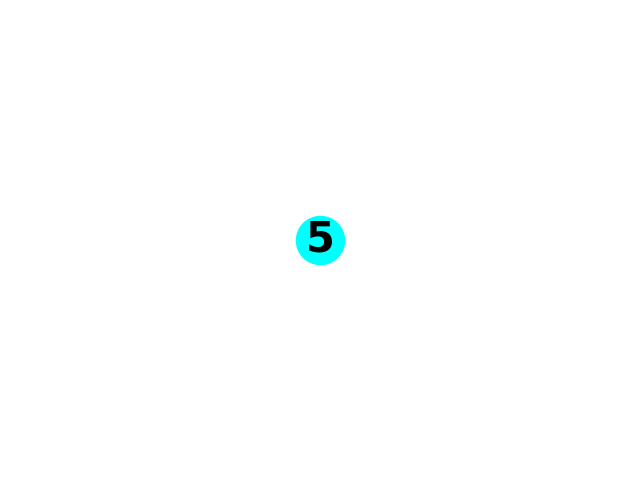}
    \end{subfigure}

    \caption{Positive-weighted free trees of weight 5}
    \label{fig:WeightedTreesFive}
\end{figure}

\begin{figure}[h]
    \centering

    \begin{subfigure}{\textwidth}
        \centering
        \includegraphics[width=0.21\textwidth,height=0.11\textheight]{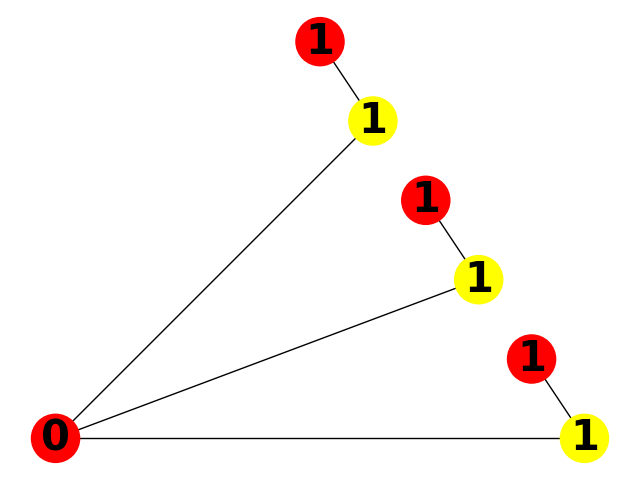}\qquad
        \includegraphics[width=0.21\textwidth,height=0.11\textheight]{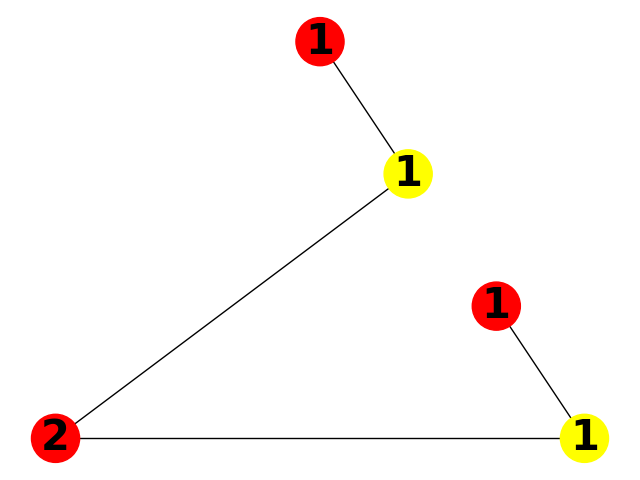}\qquad
        \includegraphics[width=0.21\textwidth,height=0.11\textheight]{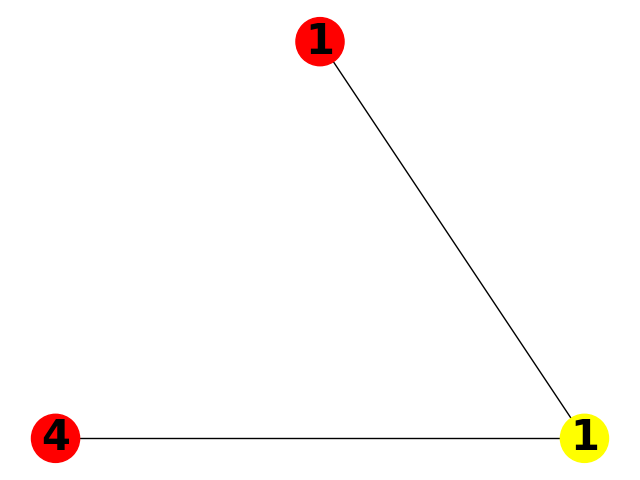}\qquad
        \includegraphics[width=0.21\textwidth,height=0.11\textheight]{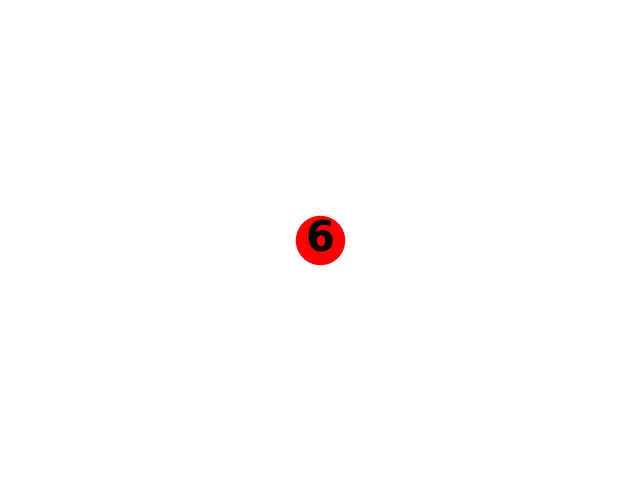}
        \caption{Block Trees with one red centroid}
        \label{fig:BlockTreesWithRedCentroid}
    \end{subfigure}

    \begin{subfigure}{\textwidth}
        \centering
        \includegraphics[width=0.25\textwidth,height=0.11\textheight]{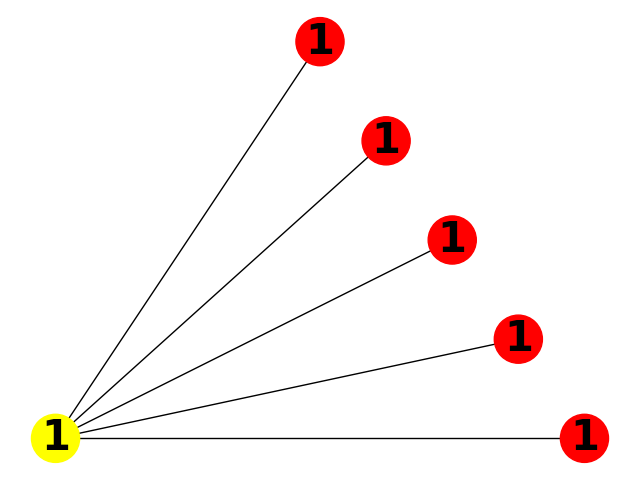}\qquad
        \includegraphics[width=0.25\textwidth,height=0.11\textheight]{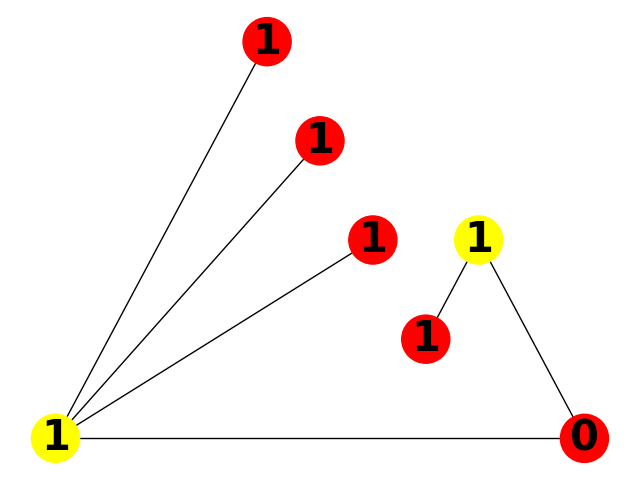}\qquad
        \includegraphics[width=0.25\textwidth,height=0.11\textheight]{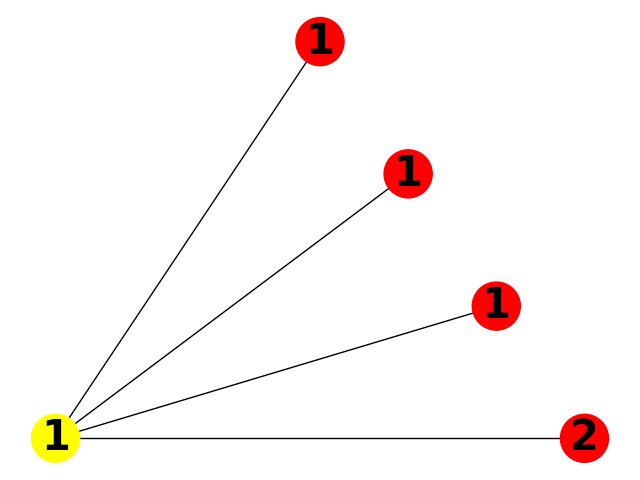}
        \includegraphics[width=0.25\textwidth,height=0.11\textheight]{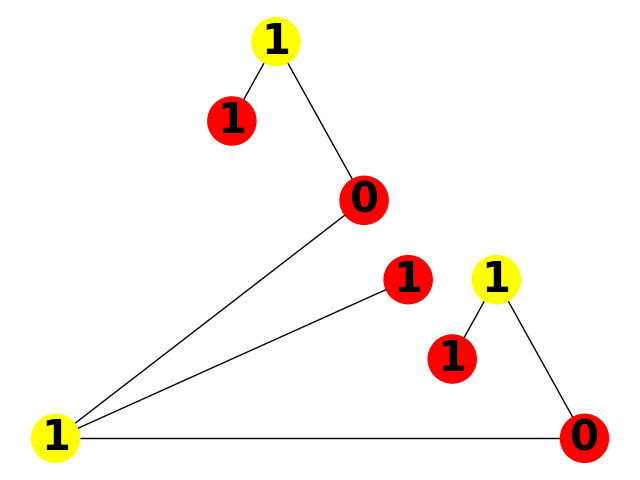}\qquad
        \includegraphics[width=0.25\textwidth,height=0.11\textheight]{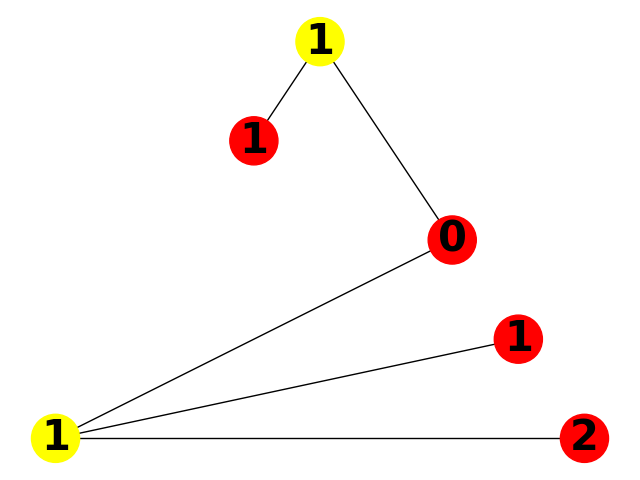}\qquad
        \includegraphics[width=0.25\textwidth,height=0.11\textheight]{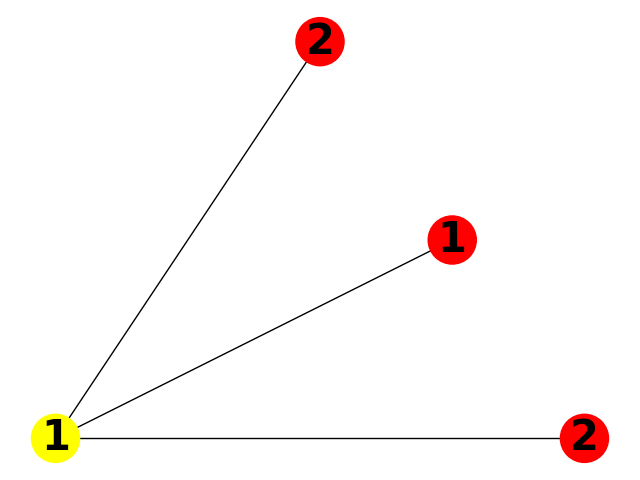}
        \caption{Block Trees with one yellow centroid}
        \label{fig:BlockTreesWithYellowCentroid}
    \end{subfigure}

    \begin{subfigure}{\textwidth}
        \centering
        \includegraphics[width=0.25\textwidth,height=0.11\textheight]{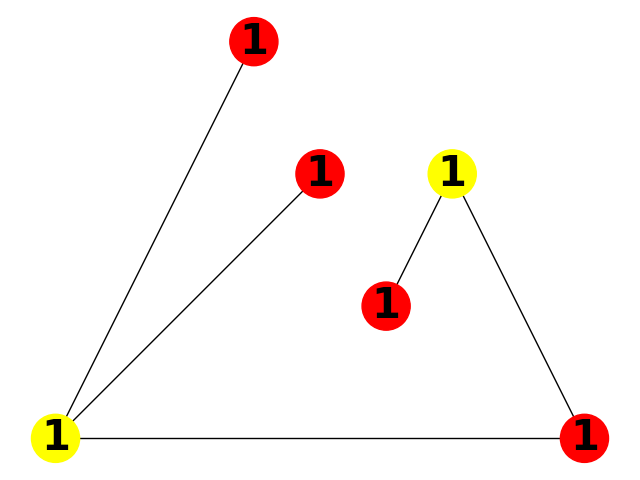}\qquad
        \includegraphics[width=0.25\textwidth,height=0.11\textheight]{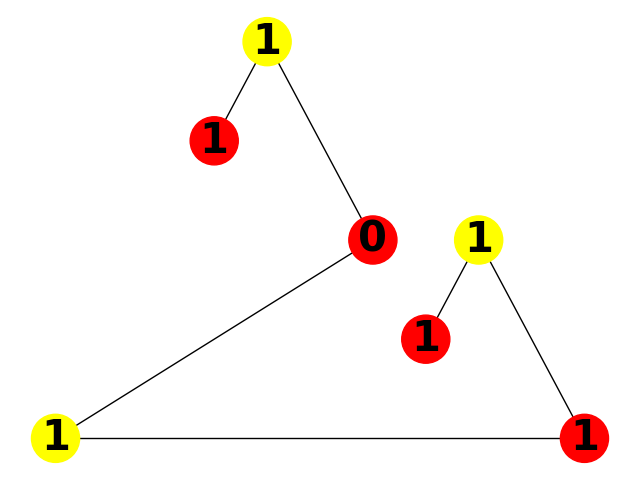}\qquad
        \includegraphics[width=0.25\textwidth,height=0.11\textheight]{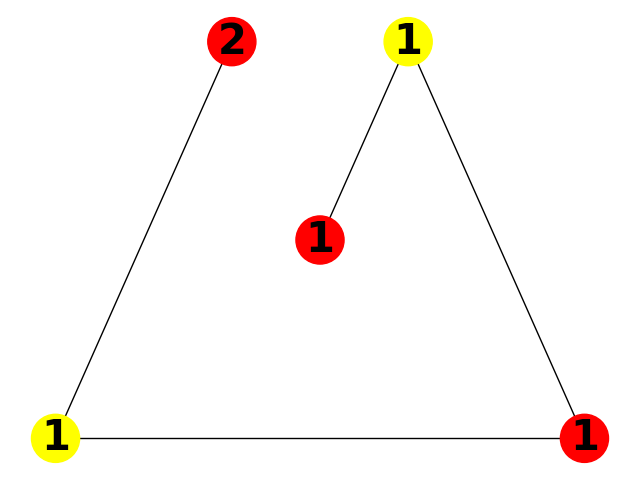}
        \includegraphics[width=0.25\textwidth,height=0.11\textheight]{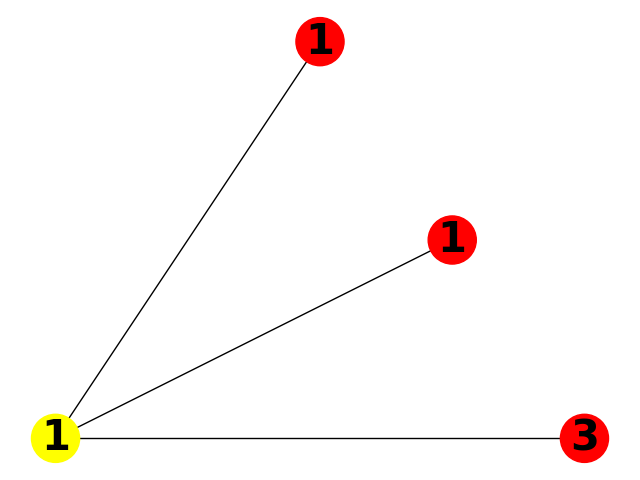}\qquad
        \includegraphics[width=0.25\textwidth,height=0.11\textheight]{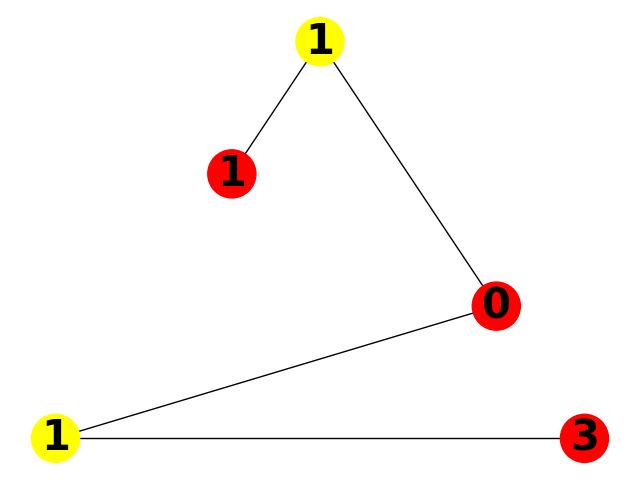}\qquad
        \includegraphics[width=0.25\textwidth,height=0.11\textheight]{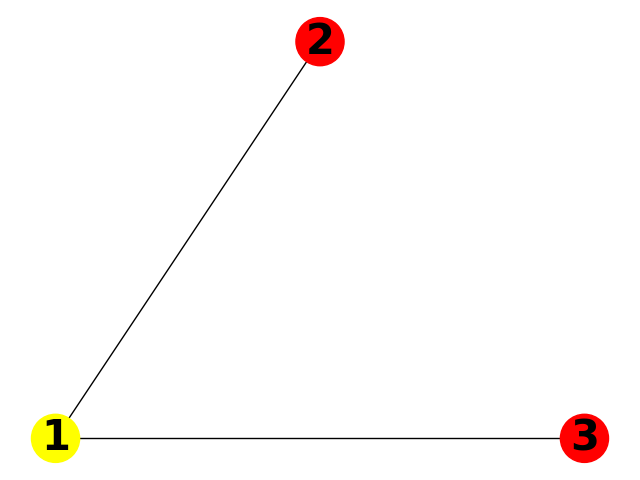}
        \caption{Block Trees with two centroids}
        \label{fig:BlockTreesWithTwoCentroids}
    \end{subfigure}

    \begin{subfigure}{\textwidth}
        \centering
        \includegraphics[width=0.25\textwidth,height=0.11\textheight]{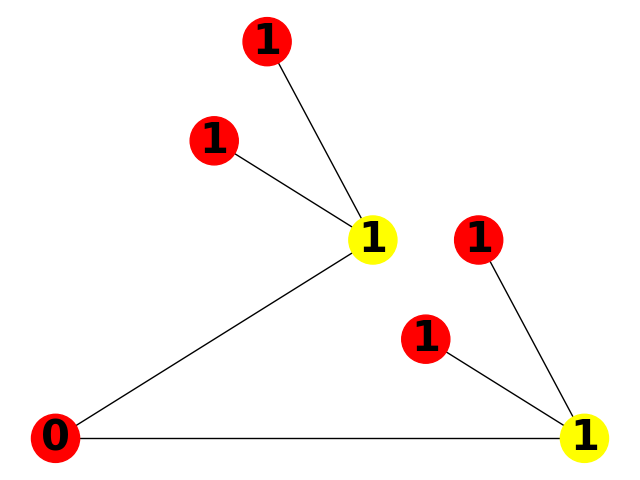}\qquad
        \includegraphics[width=0.25\textwidth,height=0.11\textheight]{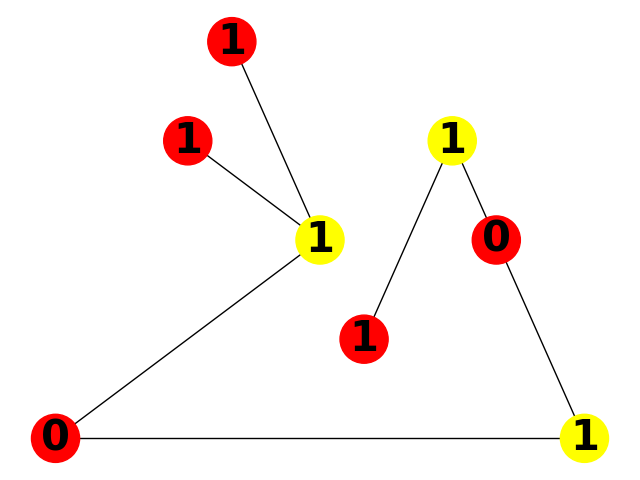}\qquad
        \includegraphics[width=0.25\textwidth,height=0.11\textheight]{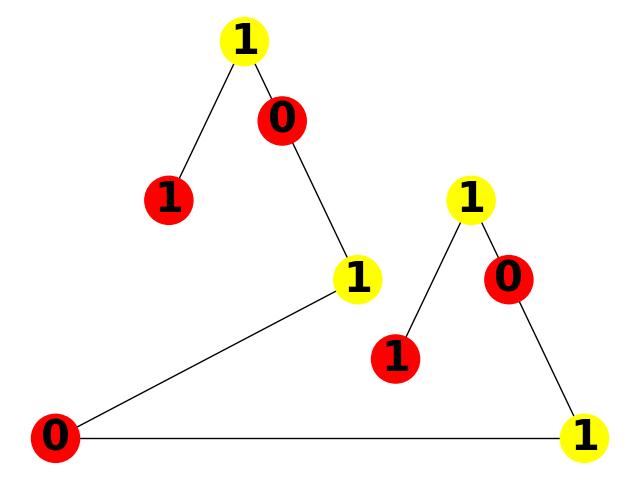}
        \includegraphics[width=0.25\textwidth,height=0.11\textheight]{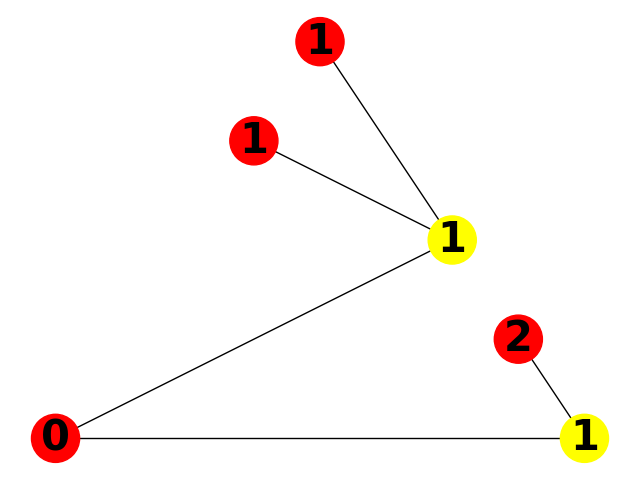}\qquad
        \includegraphics[width=0.25\textwidth,height=0.11\textheight]{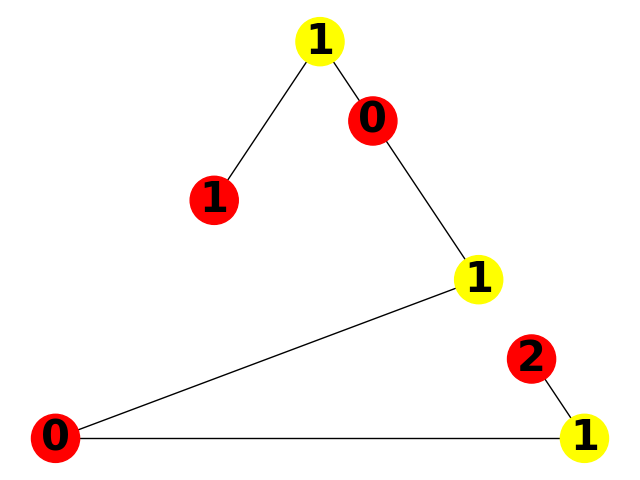}\qquad
        \includegraphics[width=0.25\textwidth,height=0.11\textheight]{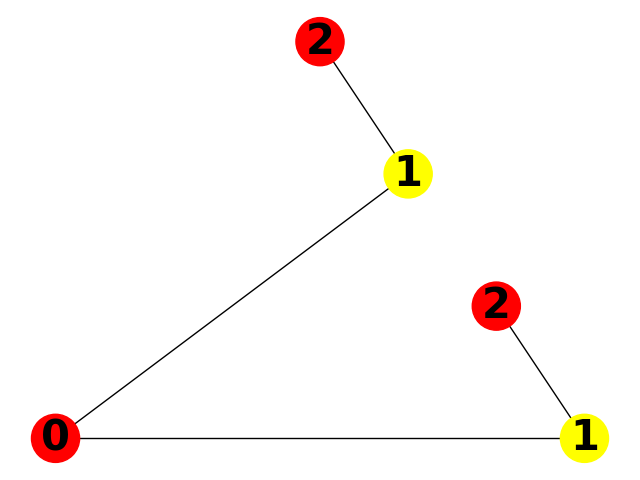}
        \caption{Block Trees with three centroids}
        \label{fig:BlockTreesWithThreeCentroids}
    \end{subfigure}

    \caption{Block Trees of Graphs on 6 vertices}
    \label{fig:BlockTreesSix}
\end{figure}

\begin{table}[h]
    \centering
    \begin{tabular}{rr}
     $n$ & Block Trees of Weight $n$\\
     \hline
     1 &                    1\\
     2 &                    1\\
     3 &                    2\\
     4 &                    4\\
     5 &                    9\\
     6 &                   22\\
     7 &                   59\\
     8 &                  165\\
     9 &                  496\\
    10 &                1,540\\
    11 &                4,960\\
    12 &               16,390\\
    13 &               55,408\\
    14 &              190,572\\
    15 &              665,699\\
    16 &            2,354,932\\
    17 &            8,424,025\\
    18 &           30,424,768\\
    19 &          110,823,984\\
    20 &          406,734,060\\
    21 &        1,502,876,903\\
    22 &        5,586,976,572\\
    23 &       20,884,546,416\\
    24 &       78,460,794,158\\
    25 &      296,124,542,120\\
    26 &    1,122,346,648,913\\
    27 &    4,270,387,848,473\\
    28 &   16,306,781,486,064\\
    29 &   62,476,518,448,854\\
    30 &  240,110,929,120,323\\
    \end{tabular}
    \caption{Number of connected block graphs with up to 30 vertices.}
    \label{tab:NumberOfBlockTrees}
\end{table}

\end{document}